\documentclass[ aps,showpacs,amsmath,amssymb,onecolumn,superscriptaddress,notitlepage,longbibliography,preprintnumbers,prx]{revtex4-2}
\usepackage{bm}
\usepackage{times}
\usepackage{multirow}
\usepackage{dcolumn}
\usepackage{tabularx}
\usepackage{graphicx}
\usepackage{subfigure}
\usepackage{ae}
\usepackage{lettrine}
\usepackage{color}
\usepackage{amsmath,amsthm,amssymb,amsfonts,boxedminipage}
\usepackage{array}
\usepackage{epstopdf}
\usepackage{booktabs}
 \usepackage{qcircuit}
\newtheorem{theorem}{Theorem}
\newtheorem{lemma}{Lemma}

\newtheorem{definition}{Definition}
\newtheorem{remark}{Remark}
\newtheorem{task}{Task}

\makeatletter
\renewcommand\@biblabel[1]{#1.}
\renewcommand{\thesubfigure}

\usepackage[noend,noline,ruled,linesnumbered]{algorithm2e}
\usepackage{hyperref}

\hypersetup{colorlinks=true, linkcolor=cyan, citecolor=cyan, urlcolor=blue }

\usepackage{braket}

\newcommand{\Ccal}{\mathcal{C}}

\newcommand{\abs}[1]{\left| #1 \right|}

\newcommand{\uwa}{Department of Physics, The University of Western Australia, Perth, WA 6009, Australia}
\newcommand{\pku}{Center on Frontiers of Computing Studies, School of Computer Science, Peking University, Beijing 100871, China}

\newcommand{\bnu}{School of Artificial Intelligence,
 Beijing Normal University, Beijing,
 100875, China}

\begin{document}

\title{Efficient Noisy Quantum State and Process Tomography}
\author{Chenyang Li}
\affiliation{\bnu}

\author{Shengxin Zhuang}
\affiliation{\uwa}

\author{Yukun Zhang}
\affiliation{\pku}

\author{Jingbo~B.~Wang}
\affiliation{\uwa}


\author{Xiao Yuan}
\email{xiaoyuan@pku.edu.cn}
\affiliation{\pku}

\author{Yusen Wu}
\email{yusen.wu@bnu.edu.cn}
\affiliation{\bnu}

\author{Chuan Wang}
\email{wangchuan@bnu.edu.cn}
\affiliation{\bnu}

\begin{abstract}

Efficiently characterizing large quantum states and processes is a central yet notoriously challenging task in quantum information science, as conventional tomography methods typically require resources that grow exponentially with system size. Here, we introduce a structure-agnostic learning framework for noisy $n$-qubit quantum circuits under~i.i.d.~single-qubit noise. We first prove that quantum states with unital noise channels admit an efficient learnable representation in the logarithmic-depth regime. We then extend this framework to quantum process tomography under constant noise, deriving a unified protocol that applies to both unital and non-unital noisy channels and retains efficient guarantees for logarithmic-depth circuits. This process-learning formulation is input-agnostic and imposes no distributional assumptions on the input quantum states.
We further study a more general regime with arbitrary noise strength. In this setting, low-weight Pauli propagation induces a terminal truncation whose threshold depends logarithmically on the inverse accuracy, leading to quasi-polynomial complexity and near-unit success probability in the average case. In contrast to the preceding two results, this arbitrary-noise guarantee does not impose any restriction on the circuit depth, and therefore covers arbitrary-depth circuits, including both the noiseless limit ($\gamma = 0$) and the strong-decoherence regime ($\gamma = \Theta(1)$).
Numerical simulations of two-dimensional Hamiltonian dynamics further demonstrate the accuracy and robustness of the approach, including for structured circuits beyond the random-circuit setting assumed in the theoretical analysis. These results provide a scalable and practically relevant route toward characterizing large-scale noisy quantum devices, addressing a key bottleneck in the development of quantum technologies.

\end{abstract}

\maketitle

\section{Introduction}
Quantum computers are entering regimes beyond the reach of classical computational power~\citep{arute2019quantum,morvan2024phase,zhong2020quantum}. Coherent manipulation of complex quantum states with hundreds of physical qubits has been demonstrated across multiple platforms, including trapped ions~\citep{smith2016many}, neutral atom arrays~\citep{evered2023high}, and superconducting qubit circuits~\citep{arute2019quantum,morvan2024phase,acharya2024quantum}. As quantum hardware continues to scale in size and complexity, the ability to characterize quantum states and quantum processes becomes critical for advancing quantum error correction code~\citep{bravyi2024high, acharya2024quantum}, quantum error mitigation~\citep{kimScalableErrorMitigation2023,o2023purification}, and quantum algorithms~\citep{kimEvidenceUtilityQuantum2023,morvan2024phase}. Among various approaches for characterizing quantum states and processes, quantum state tomography~(QST)~\citep{banaszekFocusQuantumTomography2013, blume-kohoutOptimalReliableEstimation2010, eisertQuantumCertificationBenchmarking2020, grossQuantumStateTomography2010, hradilQuantumstateEstimation1997, maurodarianoQuantumTomography2003} and quantum process tomography~(QPT)~\citep{chuangPrescriptionExperimentalDetermination1997,darianoQuantumTomographyMeasuring2001,mohseniQuantumprocessTomographyResource2008} stand as fundamental methods to reconstruct target quantum states and processes by leveraging measurement results.

The high dimensionality of Hilbert spaces imposes fundamental challenges on QST and QPT. Proven results indicate that, in the worst-case scenario, both tasks necessitate measuring a vast number of observables, incurring resource costs that grow exponentially relative to the system size~\citep{chenTightBoundsQuantum2022, haahQueryoptimalEstimationUnitary2023, haahSampleoptimalTomographyQuantum2017, odonnellEfficientQuantumTomography2016, oufkirSampleoptimalQuantumProcess2023}. Nevertheless, these ``no-go" results do not rule out efficient algorithms in the average-case scenario. Indeed, by assuming specific input state distributions (e.g., locally flat distributions) or adopting relaxed learning objectives—such as quantum mean value learning or Hamiltonian learning—QPT and QST tasks can be rendered efficient in terms of both sample and classical post-processing complexity~\citep{aaronsonEfficientTomographyNoninteracting2023, anshuSampleefficientLearningQuantum2020, arunachalamOptimalAlgorithmsLearning2023, baireyLearningLocalHamiltonian2019, cheLearningQuantumHamiltonians2021, chenQuantumAdvantagesPauli2022, cramerEfficientQuantumState2010, flammiaEfficientEstimationPauli2020, flammiaPauliErrorEstimation2021, gebhartLearningQuantumSystems2023, granadeRobustOnlineHamiltonian2012, grewalImprovedStabilizerEstimation2024, grewalLowstabilizercomplexityQuantumStates2023, grossSchurWeylDuality2021, guPracticalBlackBox2024, haahOptimalLearningQuantum2022, hangleiterRobustlyLearningHamiltonian2024, huangLearningManybodyHamiltonians2023, huangPredictingManyProperties2020, laiLearningQuantumCircuits2022, lanyonEfficientTomographyQuantum2017, liHamiltonianTomographyQuantum2020, montanaroLearningStabilizerStates2017, rouzeLearningQuantumManybody2024, stilckfrancaEfficientRobustEstimation2024, vandenbergProbabilisticErrorCancellation2023, yuRobustEfficientHamiltonian2023, zubidaOptimalShorttimeMeasurements2021,wu2025hamiltonian,wu2023quantum}. Despite this progress, efficiently characterizing quantum processes generated by noisy quantum computers remains an open problem. Specifically, a universally efficient framework capable of handling both unital and non-unital channels across arbitrary noise strengths and circuit depths has yet to be established; consequently, how to unify these disparate noise effects into a coherent learning framework remains a pivotal challenge.
On the other hand, given the power of classical artificial-intelligence methods, it is natural to consider their application to complex QPT and QST tasks, such as neural-network models~\citep{ melko2019restricted,acharya2019comparative,wannerPredictingGroundState2024,tang2024towards}, tensor networks~\citep{torlaiQuantumProcessTomography2023}, diffusion models~\citep{tang2025quadim}, and other approaches~\citep{wulearning,duEfficientLearningLinear2025}. However, these heuristic methods generally lack theoretical guarantees or may not handle QPT and QST in a noisy environment. These advances, together with the fundamental limitations discussed above, naturally raise a question: 

\emph{“Can we efficiently learn a general noisy quantum process and quantum state?”}

\begin{figure*}[t]
\begin{center}
\includegraphics[width=\linewidth]{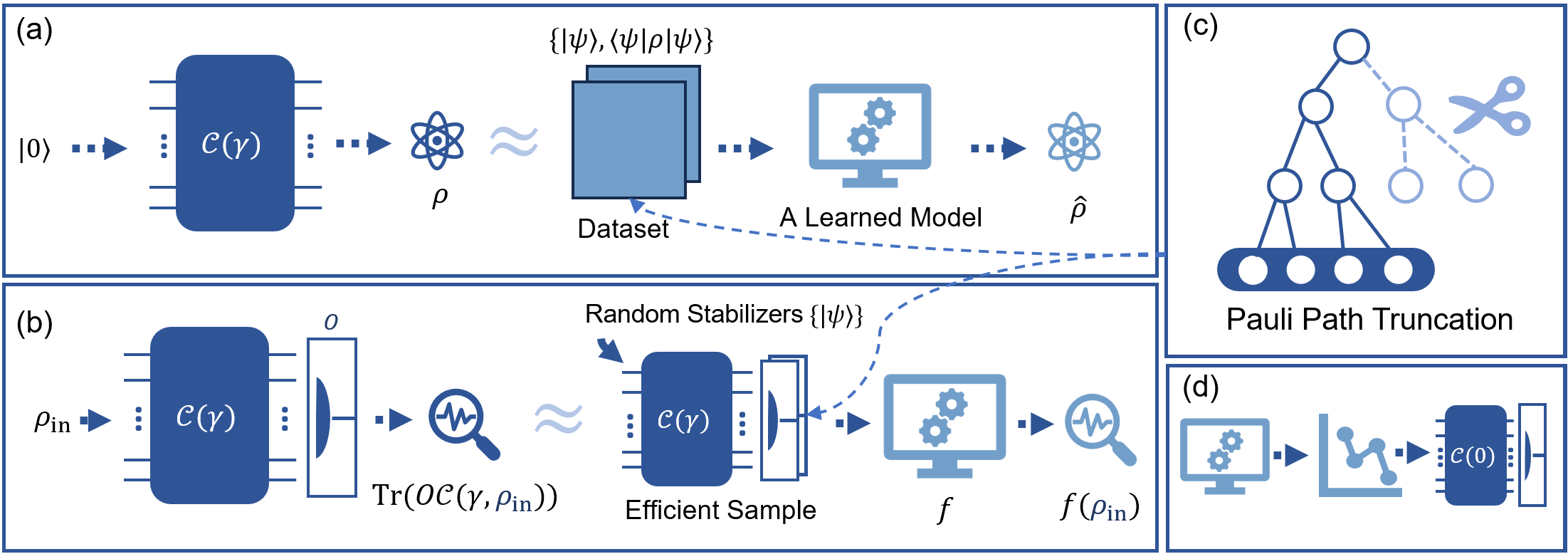}
\end{center}
\caption{(a) Illustration of the noisy quantum state learning, wherein a trained model $\hat{\rho}$ is generated by leveraging the adaptive measurement result from the target noisy quantum state $\rho$. (b) Depiction of the noisy quantum process learning. Here, the noisy quantum process $\mathcal{C}(\gamma)$ represents a $d$-depth quantum circuit with noise strength $\gamma$, and $O$ represents an unknown measurement operator. The task is to learn a function $f$ such that $\abs{f(\cdot)-{\rm Tr}[O\mathcal{C}(\gamma,\cdot)]}\leq\epsilon$ for all input quantum states $\rho_{\rm{in}}$, with efficient sample complexity. (c) Outline of the fundamental principle underlying our learning algorithm.(d) The proposed learning algorithm can be applied to the quantum error mitigation task, more details are provided in Section~\ref{sec:apply}. }
\label{fig:main}
\end{figure*}

In this paper, we answer this question by proposing a unified learning framework for both QPT and QST. Given the aforementioned ``no-go'' results for QPT and QST tasks~\citep{chenTightBoundsQuantum2022, haahQueryoptimalEstimationUnitary2023, haahSampleoptimalTomographyQuantum2017, odonnellEfficientQuantumTomography2016, oufkirSampleoptimalQuantumProcess2023}, we do not expect any efficient learning protocols in the worst-case scenario. Here, we propose an efficient learning algorithm for the average-case regime by introducing a unified representation of noisy quantum processes and states. Specifically, let $\mathcal{C}$ denote the target noisy quantum circuit. As a foundational step, we establish a provably efficient learning algorithm for unital noisy quantum state $\rho=\mathcal{C}(|0^n\rangle\langle0^n|)$ by i.i.d. local noise of constant strength admit an efficient low-weight Pauli representation in the logarithmic-depth regime. This yields a polynomial-time learning algorithm for approximating the output state in Schatten $2-$norm. 
Building upon this foundation, we  then extend the same principle to the tomography task for an constant noisy logarithmic-depth quantum process ${\rm Tr}\left(O\mathcal{C}(\cdot)\right)$ accompanied by an unknown measurement $O$, regardless of whether the underlying noise channel is unital or non-unital, can be reduced to learning an unknown observable with the decomposition $\sum_{|P|\leq\mathcal{O}(1), P\in\{I, X, Y, Z\}^{\otimes n}}\alpha_PP$, where the coefficients $\alpha_P\in\mathbb{R}$. For arbitrary local single-qubit noise strength and arbitrary circuit depth, terminal low-weight truncation gives an input-agnostic learning algorithm that is efficient at constant accuracy. 
In this regime, the effective learning space is reduced from the full $4^n$-dimensional Pauli space to a polynomial-size terminal Pauli subspace. 
This constant-accuracy efficiency persists across the full noise-strength range, from the noiseless limit $(\gamma=0)$, where local scrambling controls Pauli propagation, to constant noise levels $(\gamma=\Theta(1))$, where physical damping further suppresses high-weight contributions. 
For inverse-polynomial accuracy, direct terminal enumeration leads to quasi-polynomial complexity; for geometrically local circuits, this enumeration can be further restricted to the light cone, giving a refined coefficient count $(\mathcal O(d^D))^{l'}$.
The fundamental idea is illustrated in Fig.~\ref{fig:main}. Finally, we numerically benchmark our algorithm on noisy Hamiltonian dynamics driven by a two-dimensional lattice model~\citep{kimEvidenceUtilityQuantum2023}. The results demonstrate high accuracy for both QST and QPT tasks. 
Functionally, our work establishes a ``learning-theoretic dual'' to classical simulation frameworks~\citep{gil-fusterRelationTrainabilityDequantization2025}. While classical simulation algorithms leverage the sparsity of the Pauli path integral to compute expectation values under the white-box assumption—requiring full knowledge of circuit parameters and noise strengths~\citep{aharonovPolynomialTimeClassicalAlgorithm2023, shaoSimulatingNoisyVariational2024,schusterPolynomialtimeClassicalAlgorithm2025}—our protocol operates in a black-box setting. It reconstructs unknown processes solely from measurement data, requiring no prior information regarding the circuit architecture or specific noise levels. We reveal that the intrinsic properties of random quantum circuits directly imply a universally efficient sparse representation. Crucially, our algorithm’s complexity remains efficient across the entire noise regime $\gamma \in [0, 1)$ and is independent of circuit depth, offering a robust alternative to simulation methods that are often restricted to the deep-circuit or high-noise limits. Beyond being a dual to simulation, our approach significantly advances existing learning paradigms~\citep{huangLearningPredictArbitrary2023, chenPredictingQuantumChannels2024,razaOnlineLearningQuantum2024,crupiEfficientCharacterizationCoherent2025}. Unlike previous works that rely on specific input state distributions, unital dynamics or exponential scaling, our algorithm achieves input-agnostic prediction efficiently, remaining applicable to arbitrary input states and non-unital dynamics by leveraging the intrinsic sparse representation. Finally, we demonstrate the practical utility of this framework by applying the learned process models to quantum error mitigation, providing a pathway for noise-resilient quantum computing beyond mere verification. 

\section{Preliminaries and notation}
\label{sec:preliminaries}
To motivate and contextualize our contribution, we briefly review the requisite background on noisy quantum channels and circuits.
\begin{definition}[Single-Qubit Pauli Channel]
    Let  $\mathcal{E}_{\rm Pauli}$ denote the single-qubit Pauli channel, which is
\begin{equation}
    \mathcal{E}_{{\rm Pauli}}(\rho) =\gamma_1\rho+\gamma_2X\rho X^{\dagger}+\gamma_3Y\rho Y^{\dagger}+\gamma_4Z\rho Z^{\dagger},
\end{equation}
where real parameters $\gamma_1+\gamma_2+\gamma_3+\gamma_4 = 1$, and $\gamma_i\in[0,1]$ for $i\in[4]$. 
\end{definition}
As a standard unital quantum channel, the Pauli noise has the property $\mathcal{E}_{{\rm Pauli}}(I)=I$, $\mathcal{E}_{{\rm Pauli}}(X)=(1-2(\gamma_3+\gamma_4))X$, $\mathcal{E}_{{\rm Pauli}}(Y)=(1-2(\gamma_2+\gamma_4))Y$ and $\mathcal{E}_{{\rm Pauli}}(Z)=(1-2(\gamma_2+\gamma_3))Z$. Note that if $\gamma_2=\gamma_3=\gamma_4$, $\mathcal{E}$ degenerates to an i.i.d. single-qubit depolarizing noise, which is $\mathcal{E}_{{\rm depo}}(P)=(1-\gamma)P$ for $P\in \{X,Y,Z\}$. Techniques like Pauli twirling are employed to transform complex unital channels into diagonal forms on the Pauli basis~\citep{chen2023learnability, wallman2016noise}. In the following, we utilize the Pauli noise channel to represent the unital channel.

Another widely studied class of quantum channels is the non-unital channel, which describes channels that do not map the identity operator to itself. This kind of noise often reflects complicated environmental disturbances on the quantum system, where a canonical example is the amplitude damping. Ref.~\citep{angrisaniSimulatingQuantumCircuits2025} decompose the normal form of a non-unital single-qubit noise channel $\mathcal{E}$ as
\begin{equation}
    \begin{aligned}
        \mathcal{E} = \mathcal{E}_{\rm depo}^{\gamma}  \circ\mathcal{E}',
    \end{aligned}
    \label{Eq: nonunitaldecompose}
\end{equation}
where $\mathcal{E}'$ is a suitable (non-physical) linear map and $\mathcal{E}_{\rm depo}^\gamma$ is a depolarizing noise with the effective depolarizing rate $\gamma$. $\mathcal{E} $ is characterized by contraction parameters \(D=(D_X,D_Y,D_Z)\) and a translation vector \(t=(t_X,t_Y,t_Z)\)\citep{angrisaniSimulatingQuantumCircuits2025}. This normal form includes depolarizing-like, dephasing-like, and non-unital noise, including amplitude damping. Throughout, ``arbitrary local noise” refers to gate-independent tensor-product single-qubit noise channels; it does not include coherent unitary errors, correlated noise, or gate-dependent/non-Markovian noise. Given this observation, we define a unified noise parameter across unital and non-unital noise channels:
\begin{equation}
\gamma = \left\{
    \begin{array}{ll}
    2(\gamma_i+\gamma_j) ~~ (i,j)\in\{2,3,4\}, & \mathcal{E}~~\text{is unital}   \\
    1-\chi_{\mathcal{D}}(\mathcal{E}), & \mathcal{E}~\text{is non-unital} 
    \end{array}\right.
\end{equation}
where $\chi_{\mathcal{D}}(\mathcal{E})$ denotes the mean squared contraction coefficient of $\mathcal{E}$ with respect to the locally unbiased distribution $\mathcal{D}$. 
The details of the non-unital noise are in Appendix~\ref{sec:nu}. 


\begin{definition}[Schatten $\tau-$Norm]
\label{def:L1 norm}
   The Schatten $\tau-$norm of a matrix $A$ is defined as $\|A\|_\tau = (\sum_i\nu_i^\tau) ^{\frac 1 \tau}$, where $\nu_i$ is the singular value of $A$ and $\tau$ is a positive integer. Note that $\|A\|_1={\rm{Tr}}[\sqrt{AA^\dagger}]$.

\end{definition}
\begin{definition}[The Squared Normalized Frobenius Norm]
\label{def:Fro norm}
Suppose the matrix $A = \sum_{P}\alpha_P P$, with $P\in\{X,Y,Z,I\}^{\otimes n}$, its squared normalized Frobenius norm is defined by $\|A\|_F^2= \sum_{P}\alpha_P^2$.
\end{definition}

\begin{definition}[Hamming Weight of Pauli Operators]
\label{def: hamming weight}
Suppose $P$ represents an $n$-qubit (normalized) Pauli operator, then its Hamming weight $\abs{P}$ is defined as the number of qubits that are non-trivially acted by $P$.
\end{definition}


\section{Noisy Quantum Circuit Setup}

In this paper, we consider an $n$-qubit noisy quantum process:
\begin{align}
\mathcal{C}=\mathcal{E}^{\otimes n}\mathcal{C}_d\mathcal{E}^{\otimes n}\mathcal{C}_{d-1}\cdots\mathcal{E}^{\otimes n}\mathcal{C}_1,
\label{Equ:noisychannel}
\end{align}
where a local noise channel $\mathcal{E}$ with noise strength $\gamma$ is applied uniformly throughout the circuit. We consider unital noise in the QST setting and both unital and non-unital noise in the QPT setting. The circuit has depth $d$ and each layer $\mathcal{C}_i$ comprises non-overlapping two-qubit gates acting on arbitrary pairs, with each gate sampled uniformly from a local 2-design unitary group. Crucially, our learning protocol is \emph{agnostic to circuit connectivity and applicable to arbitrary architectures}, including standard parameterized brickwork structures, quantum cellular automata, and tree-structured circuits~\citep{zhouQuantumApproximateOptimization2020, kandalaHardwareefficientVariationalQuantum2017, grimsleyAdaptiveVariationalAlgorithm2019,benedettiParameterizedQuantumCircuits2019}, provided that each two-qubit gate satisfies the local 2-design assumption. By interleaving unitary layers $\mathcal{C}_i$ with i.i.d. tensor-product noise $\mathcal{E}^{\otimes n}$, this framework serves as a canonical model for quantum verification. Specifically, it provides the theoretical foundation for validating tasks that are proven to be classically hard, such as quantum advantage experiments~\citep{arute2019quantum, kimEvidenceUtilityQuantum2023}. A rigorous definition of the model's topological generality is provided in the Appendix~\ref{sec:introC}.

\vspace{8px}

\section{Noisy Quantum State Tomography}
As a state-level warm-up, we first consider the noisy QST task in the context of unital-noise environment under the constant noise strength, where the objective is to construct a classical approximation $\hat{\rho}$ of the output state $\rho = \mathcal{C}(|0^n\rangle\langle 0^n|)$ such that $\|\rho - \hat{\rho}\|_2 \le \epsilon$. In this regime, we expand the noisy state $\mathcal{C}(|0^n\rangle\langle 0^n|)$ in the Pauli basis: $\mathcal{C}(|0^n\rangle\langle 0^n|)=\sum_{P\in\{I,X,Y,Z\}^{\otimes n}}\alpha_P P$, where $\alpha_P=2^{-n}\text{Tr}[P\mathcal{C}(|0^n\rangle\langle 0^n|)]$.
The central point is that unital noise suppresses the high-weight terminal Pauli components. Therefore it is enough to learn a truncated low-weight approximation $\hat{\rho}=
    \sum_{|P|\leq l'}\alpha_P P$.
    
To make the truncation mechanism explicit, we use the following Pauli-path representation, whose proof is given in Appendix~\ref{sec:pauli}.
\begin{lemma}[Unified representation of noisy quantum states]
\label{le:noise_main}
Let $\rho=\mathcal C(|0^n\rangle\langle0^n|)$, where
$\mathcal C=\mathcal E^{\otimes n}\mathcal C_d\mathcal E^{\otimes n}\mathcal C_{d-1}\cdots\mathcal E^{\otimes n}\mathcal C_1$ is a $d$-depth noisy circuit, each $\mathcal C_i(\cdot)=C_i(\cdot)C_i^\dagger$ is a two-qubit random circuit layer, and $\mathcal E$ is a unital single-qubit noise channel with strength $\gamma$. Then
\begin{equation}
    \rho=
    \sum_{s\in\tilde{\mathcal P}_n^{\otimes(d+1)}}
    (1-\gamma)^{|s|}\Phi(\mathcal C,s)s_d,
\end{equation}
where $s=s_0s_1\cdots s_d$ is a Pauli path in the normalized Pauli basis $\tilde{\mathcal P}_n=\{I/\sqrt2,X/\sqrt2,Y/\sqrt2,Z/\sqrt2\}^{\otimes n}$, $|s|$ is its total Pauli weight, and
\begin{equation}
    \Phi(\mathcal C,s)=
    {\rm Tr}\!\left(s_d\mathcal C_d(s_{d-1})\right)
    \cdots
    {\rm Tr}\!\left(s_1\mathcal C_1(s_0)\right)
    \langle0^n|s_0|0^n\rangle .
\end{equation}
\end{lemma}

The fundamental idea is to obtain an efficient representation of a noisy quantum state by leveraging Lemma~\ref{le:noise_main}. The contribution of each Pauli path $s_d$ is determined by a related pre-factor $(1-\gamma)^{|s|}\Phi(\mathcal{C},s)$, which decays exponentially with the Pauli-path weight $\abs{s}$. Since $\abs{s_d}\leq\abs{s}$, we truncate the noisy-state representation in Lemma~\ref{le:noise_main} to terms with $\abs{s_d}\leq l^{\prime}$. It therefore suffices to show that the rest of the average-case error $\mathbb{E}_{\mathcal{C}}[\sum_{\abs{s_d}> l'}\Phi(\mathcal{C},s)]^2$ is a constant due to that the $\|\rho\|_2=1$. 

\begin{lemma}[Unital noisy state truncation]
\label{le:state_truncation_main}
Let $\rho=\mathcal C(|0^n\rangle\langle0^n|)$ be generated by Eq.~\ref{Equ:noisychannel} with unital i.i.d. single-qubit noise of strength $\gamma$. With probability at least $1-\delta$ over the random circuit ensemble, there exists a truncated estimator
$\hat{\rho}=\sum_{P\in\mathcal T_{l'}}\alpha_P P$,
such that $\|\rho- \hat{\rho}\|_2\leq\epsilon$, provided
\begin{equation}
    l'=\mathcal O\!\left(\log(\epsilon^{-1}\delta^{-1})\right)
\end{equation}
in the constant-noise regime.
\end{lemma}

Meanwhile the $\mathbb{E}_{\mathcal{C}}[\Phi(\mathcal{C},s)]^2$ exists a lower bound $15^{-|s|}$ implies that the number of non-trivial terms (those with $\alpha_{s_{d}}\neq 0$) is bounded by $N_P \leq 2^{\mathcal{O}(l)}=\mathcal{O}(1/\epsilon_1)$. To guarantee the Pauli path contribution is not zero, $|s|=0$ or $|s|\geq d+1$, which indicates that $l'=l-d$. Hence we get that when $d=\log n$, and the required accuracy is $\epsilon_1=1/{\rm poly}(n)$, all Pauli terms $s_d$ appearing in the ansatz $\hat{\rho}$ can be enumerated efficiently. Consequently, tomography of the noisy state $\rho$ is reduced to tomography of its approximation $\hat{\rho}$, determining the unknown coefficients $\alpha_{s_{d}}$ for $s_d$, then it suffices to perform the noisy-state tomography task. Since all `low-weight' Pauli operators $s_d$ can be enumerated in advance, the classical shadow method~\citep{huangPredictingManyProperties2020} is a natural candidate for estimating the coefficients $\alpha_{s_{d}}$, yielding an $\mathcal{O}\left(\log(1/\epsilon_1)\epsilon_1^{-2}\right)$ sample-complexity guarantee.

Nevertheless, the classical shadow method may not extend directly to quantum process tomography tasks. To implement a `unified' learning approach for both quantum noisy state and process tomography tasks, we provide another method for estimating coefficients $\alpha_{s_{d}}$ from the quantum randomized measurement results. We generate a dataset $\{\ket{\psi_j}=\otimes^{n}_{i=1}\ket{\psi_{i,j}},v_j=\langle\psi_j|\rho|\psi_j\rangle\}_{j=1}^{N_{\rm data}}$ by drawing each single-qubit stabilizer $\ket{\psi_{i,j}}$ uniformly sampled from the set ${\rm Stab}=\left\{\ket{0}, \ket{1}, \ket{+},\ket{-},\ket{y+},\ket{y-} \right\}$. Here, the quantum state overlap $v_j=\langle\psi_j|\rho|\psi_j\rangle$ can be efficiently obtained by using the SWAP-test method~\citep{buhrman2001quantum}. Without loss of generality, we assume each single-qubit stabilizer state can be prepared by $|\psi_{i,j}\rangle=U_{i,j}|0\rangle_i$, where $U_{i,j}$ is a random single-qubit Clifford gate. By leveraging the orthogonal property of single-qubit Pauli operators $Q_i$ in the context of the Clifford ensemble, that is
\begin{equation}
\label{Eq:cliffordproperty}
    \mathbb{E}_{U_{i,j}\sim {\rm Cl}(2)}\left[U_{i,j}^{\dagger\otimes 2}(Q_i\otimes Q_i')U_{i,j}^{\otimes 2}\right]=
    \left\{
    \begin{aligned}
       & I^{\otimes 2}, \  \text{if} \ Q_i=Q_i'=I,\\
       & \frac{1}{3}\sum_{Q_i\in\{X,Y,Z\}^{\otimes 2}}\left(Q_i\otimes Q_i\right), \ \text{if} \ Q_i=Q_i'\neq I.\\
       & 0, \ \text{if} \ Q_i\neq Q_i',
        \end{aligned}
    \right.
\end{equation}
coefficients $\alpha_{s_d}$ are obtained as
\begin{equation}
    \begin{aligned}
        \alpha_{s_d} =3^{|s_d|}\mathbb{E}_{\ket{\psi_j}\sim {\rm Stab}^{\otimes n}}v_j\bra{\psi_j}s_d\ket{\psi_j}
        \approx\frac{3^{|s_d|}}{N_{\rm data}}\sum_{j=1}^{N_{\rm data}}v_j\bra{\psi_j}s_d\ket{\psi_j}.
    \end{aligned}
\end{equation}
The in-depth explanation of the learning Algorithm is provided in Appendix~\ref{sec:sla}. We note that the above learning approach is efficient in both sample and computational complexity (classical post-processing).

\begin{algorithm}
\caption{Quantum State Learning Algorithm}
\label{alg:qsl}
\textbf{Input:} Product stabilizer states $\{|\psi_j\rangle\}_{j=1}^{N_{\rm data}}$ and accuracy parameter $\epsilon$;

\textbf{Output:} An estimator $\hat\rho$ of $\rho=\mathcal C(|0^n\rangle\langle0^n|)$;

Let $l'$ be the state truncation threshold and enumerate $\mathcal T_{l'}=\{P\in\mathcal P_n:|P|\leq l'\}$;

\textbf{For} $j\in[N_{\rm data}]$:

\quad Obtain $v_j=\langle\psi_j|\rho|\psi_j\rangle$ by a SWAP test or an equivalent randomized-measurement routine;

\textbf{End For}

\textbf{For} each $P\in\mathcal T_{l'}$:

\quad Compute $\hat\alpha_P=\frac{3^{|P|}}{N_{\rm data}}\sum_{j=1}^{N_{\rm data}}v_j\langle\psi_j|P|\psi_j\rangle$;

\textbf{End For}

\textbf{Output}: $\hat\rho=\sum_{P\in\mathcal T_{l'}}\hat\alpha_P P$.
\end{algorithm}
By the Algorithm~\ref{alg:qsl} and Hoeffding's inequality, we obtain the following theoretical guarantee.
\begin{theorem}[Noisy Quantum State Learning]
\label{the:sl}
For any unital noisy quantum state $\rho$ prepared by a noisy quantum circuit $\mathcal{C}$ (Eq.~\ref{Equ:noisychannel}) with constant noise strength and $d=\mathcal{O}(\log n)$, where $\mathcal{C}_i$ is a layer of two-qubit random haar quantum gates, there exists a learning algorithm that can efficiently solve $\|\rho-\hat{\rho}\|_2\leq \epsilon$  with success probability $\geq 1-\delta$. The sample complexity and classical post-processing complexity ${\rm poly}(n, 1/\epsilon,1/\delta)$.
\end{theorem}
Theorem~\ref{the:sl}  highlights the role of noise in state learning: the noise strengthens the constraints on Pauli weights, reducing the number of effective coefficients. The learned object is a reusable low-weight Pauli representation of the noisy state rather than a full dense matrix. Since the guarantee is in Schatten-$2$ norm, its direct applications are limited to probes stable under Hilbert--Schmidt perturbations:
\begin{equation}
    \left|{\rm Tr}\!\left[A(\rho-\hat\rho)\right]\right|
    \leq
    \|A\|_2\,\|\rho-\hat\rho\|_2 ,
\end{equation}
and therefore supports overlap estimation with known reference states and classical post-processing of $\mathcal{O}(1)$-local observables whose Pauli coefficient vector has controlled $\ell_2$ norm. Numerically, we observe that the approximation also performs well under the trace norm; these applications are detailed in Appendix~\ref{sec:qst-applications}.

\section{Noisy Quantum Process Tomography}

Now we generalize above analysis and learning method to QPT tasks. Given an $n$-qubit observable $O = \sum_{k=1}^M c_k Q_k$ composed of $M = \text{poly}(n)$ Pauli terms, the goal is to learn a function $f$ from measurement results. This function predicts the expectation value for arbitrary valid input state $\rho_{\text{in}}$, satisfying $\abs{f(\rho_{\text{in}}) - \text{Tr}[O \mathcal{C}(\rho_{\text{in}})]} \le \epsilon$. Here, we do not restrict the locality of $O$, actually, it may include global Pauli terms with $\abs{P_k}=\mathcal{O}(n)$. Here, We do not require full knowledge of the process $\mathcal C$, but assume access to measurement outcomes of $O$.

Inspired by the aforementioned noisy QST task, we propose an efficient learning algorithm for the QPT task. Remarkably, we find our learning algorithm is even applicable to non-unital noise channels in the QPT tasks. In detail, we focus on the Heisenberg picture from the perspective of the operator evolution, transforming the objective from learning the quantum channel ${\rm Tr}[O\mathcal{C}(\cdot)]$ to learn the dual operator $\mathcal{C}^\dagger(O)$, defined via the duality:
\begin{equation}
    {\rm Tr}\left[\mathcal{C}(\cdot)O\right]= {\rm Tr}\left[~\cdot~\mathcal{C}^{\dagger}(O)\right].
\end{equation}

Similar to the QST task, our fundamental strategy relies on finding an effective representation $\sum_{|P|\leq l^{\prime}}\beta_PP$, which approximates the target $\mathcal{C}^{\dagger}(O)$ in terms of the normalized Frobenius norm. To ensure independence from the input state distribution, the core idea is to introduce a ``decay coefficient" such that $\beta_P \propto c^{|P|}$ for some constant $c < 1$. 
Assuming the input states are drawn from a structured ensemble, such as locally flat or general product distributions, Refs.~\cite{huangLearningPredictArbitrary2023, chenPredictingQuantumChannels2024} propose learning a Hermitian operator by truncating all high-weight Pauli terms. This approach yields a learning algorithm with a sample complexity of $\mathcal{O}(n^{\log(1/\epsilon)})$, which crucially depends on the specific input-state ensemble. To decouple learning efficiency from the input state distribution, Ref.~\cite{razaOnlineLearningQuantum2024} leverages an online learning framework with a sample complexity of $\tilde{\mathcal{O}}(\sqrt{n})$. However, it incurs an $\mathcal{O}(4^n)$ classical post-processing overhead to update and recover the associated Choi matrix. This exponentially increasing runtime limits its scalability for large-scale systems. A complementary approach is the low-degree noise tomography protocol of Ref.~\cite{crupiEfficientCharacterizationCoherent2025}, which reconstructs a gate layer followed by low-degree noise using randomized Pauli product-state preparations and measurements, with the inverse gate layer handled in classical post-processing. Its efficiency, however, relies on a fixed low-degree noise ansatz, whereas i.i.d. single-qubit noise can have process-matrix support up to Pauli weight $n$. Here, we overcome these barriers by exploiting the sparse terminal Pauli representation induced by noisy random circuits. Under constant noise, the legal-path structure yields the sharper count summarized as our Result 1 in Table~\ref{tab:comparison}; for arbitrary noise strength, the universal low-weight terminal enumeration gives the input-agnostic and depth-independent guarantee summarized as our Result 2. In both cases, the learned effective observable predicts outcomes for arbitrary input states, including highly entangled states that explicitly violate the distributional assumptions of prior work. A more detailed comparison is provided in Table~\ref{tab:comparison}.

\subsection{The Noisy Process with Constant Noise strength}
Here we start with the similar requirement of QST tasks, which demands the noise strength is $\Omega(1)$.

For unital noise channel, the analysis is the same to QPT tasks.Regarding the non-unital noise, we leverage the decomposition technique which factorizes the channel into a depolarizing component and a suitable (non-physical) linear map $\mathcal{E}'$. Crucially, the adjoint of this effective map is unital ($\mathcal{E}^{\prime\dagger}(I)=I$)\citep{angrisaniSimulatingQuantumCircuits2025}, ensuring that the operator evolution remains bounded in the normalized Frobenius norm. This structural property facilitates rigorous truncation bounds even for non-unital physical noise, ensuring that $\mathcal{C}^\dagger(O)$ retains a sparse Pauli structure amenable to efficient learning. We derive a generalized scaling law for the process truncation: 

\begin{lemma}[Constant  Noisy Process Truncation]
\label{lem:unified_truncation}
Consider a $d$-depth noisy quantum circuit $\mathcal{C}$ comprising layers of two-qubit Haar-random gates interleaved with i.i.d. single-qubit noise channels $\mathcal{E}$ of constant strength. For any observable $O$ and input state $\rho_{\text{in}}$, there exists a truncated adjoint operator $\mathcal{C}^{(l')\dagger}(O)$ with Pauli weight at most $l'$ , such that 
\begin{equation}
    \left| \operatorname{Tr} \left[ \rho_{\text{in}} \left( \mathcal{C}^{(l')\dagger}(O) - \mathcal{C}^{\dagger}(O) \right) \right] \right| \le \epsilon
\end{equation}
holds with probability at least $1-\delta$. The sufficient truncation weight $l'$ is given by
\begin{equation}
    l' = \mathcal{O} \left( \log(1/ \epsilon \delta ) \right).
\end{equation}
\end{lemma}

The detailed proof is provided in the Appendix~\ref{sec:constant noise process}. 
Similar to the noisy QST task, reconstructing $\mathcal{C}^{(l')\dagger}(O)$ proceeds from the data set $\mathcal{D}_{\rm QPT}=\left\{\ket{\psi_j}=\otimes^{n}_{i=1}\ket{\psi_{i,j}}, \phi_j = {\rm Tr}\left[O\mathcal{C}(|\psi_j\rangle\langle\psi_j|)\right]\right\}_{j=1}^{N_{\rm data}},$ where $\ket{\psi_{i,j}}$ is a single-qubit stabilizer randomly sampled from the set ${\rm Stab}$, and $\phi_{j}$ denotes the output of the target quantum process. Thus, coefficients $\beta_P$ can be learned efficiently via    $\beta_P =\frac{3^{|P|}}{N_{\rm data}}\sum_{i=1}^{N_{\rm data}}\phi_j\bra{\psi_i}P\ket{\psi_i}.$
\begin{theorem}[Constant Noisy Process Learning Complexity]
For any noisy quantum process $\Ccal$ given by Eq.~\ref{Equ:noisychannel} with constant noise strength, and circuit depth $d=\mathcal{O}(\log n)$ and an $n$-qubit observable $O=\sum_{k=1}^M c_k Q_k$ with $M={\rm poly}(n)$, there exists a learning algorithm that learns a function $f$ from measurement results, satisfying $\abs{f(\rho_{\text{in}}) - {\rm Tr}[O \mathcal{C}(\rho_{\text{in}})]} \le \epsilon$ with success probability $\geq 1-\delta$. The sample complexity and the classical post-processing complexity are ${\rm poly}(n,1/\epsilon,1/\delta)$.
    
\end{theorem}

From the algorithm above, it can be observed that the computational overhead primarily stems from two sources: (1) the sampling complexity $N_{\rm data}$, (2) and the complexity of classical post-processing. Both of these costs depend on the number of $P$, which in turn is governed by how many legal Pauli paths are retained, in other words, the number of Pauli operators $P$ with non-zero parameter $\beta_{P}$ contained in $\mathcal{C}^{(l')\dagger}(O)$, where the weight of $P$ is given by $l^\prime$. Therefore, a rough estimate of the number of legal paths $N_P$ is $\mathcal{O}(n^{l^\prime})$. For $\epsilon=1/n$, the number of legal paths becomes $\mathcal{O}(n^{\log n})$, incurring quasi-polynomial sampling and post-processing complexity. 

However, we can tighten the bound to retain only $e^{\mathcal{O}(l^\prime)}$ legal paths, where the non-identity positions in subsequent layers are strictly determined by the preceding layer. Starting from an observable $O$ consisting of $M$ terms, the number of valid Pauli combinations is sharply restricted to at most $M \cdot e^{\mathcal{O}(l)}$. Still considering the effective depth, thus we have that when the noise strength is a constant, the Pauli combination of the legal paths is at most $M2^{\mathcal{O}(l')}$ with depth $d=\mathcal{O}(\log n)$. If $d>\mathcal{O}(\log n)$, one can directly output the zero function, which is a good approximation of the noisy circuit. As a result, when $l' = \mathcal{O}(\log n)$, the exponential growth $e^{\mathcal{O}(l')}$ collapses into a standard polynomial scaling $\text{poly}(n)$, ensuring the overall algorithm remains highly efficient. 
\subsection{Arbitrary Noisy Process}
While the constant-noise analysis reveals the intrinsic sparsity $M e^{\mathcal O(l')}$, it relies on the path-damping mechanism and the associated logarithmic-depth regime. We next present an algorithm that applies to entire noise regime $\gamma\in[0,1)$ and arbitrary depth. 

Instead of only relying on the noise suppression, we introduce the property of locally scrambling, which is $\mathbb{E}[{\rm Tr}(UPU^\dagger\rho)]^2\leq (2/3)^{|P|}$. Existing low-weight Pauli-propagation bounds control the mean-square contribution of paths that leave a low-weight sector~\cite{angrisaniClassicallyEstimatingObservables2025}. In contrast, our algorithm does not enumerate such internal paths; it directly estimates the terminal Pauli coefficients of $\mathcal{C}^\dagger(O)$. We define the terminal low-weight set as $\mathcal T_{l'}:=\{P\in \mathcal P_n:\ |P|\le l'\}$. In the Appendix~\ref{sec: arbitrary noise process}, we show that the set of Pauli paths discarded by terminal low-weight truncation is contained in the set of paths discarded by the low-weight path truncation. Consequently, the terminal truncation error inherits the same average-case second-moment bound, yielding a depth-independent pointwise prediction guarantee for an arbitrary input state. The key idea is derived from the observation that the support of any Pauli string generated
by conjugating $P$ through one layer gate $U_t$ is contained in the union of gate pairs
intersecting $\operatorname{supp}(P)$, which has size at most $2|\operatorname{supp}(P)| = 2|P|$. By learning all terminal coefficients in this set,
our algorithm captures the entire contribution of the retained path sector
without explicitly enumerating internal paths.

\begin{lemma}[Arbitrary Noisy Process Truncation]
Consider a $d$-depth noisy quantum circuit $\mathcal{C}$ comprising layers of two-qubit Haar-random gates interleaved with i.i.d. single-qubit noise channels $\mathcal{E}$ of arbitrary strength. For any observable $O$ and input state $\rho_{\text{in}}$, there exists a truncated adjoint operator $\mathcal{C}^{(l')\dagger}(O)$ with Pauli weight at most $l'$ , such that 
\begin{equation}
    \left| \operatorname{Tr} \left[ \rho_{\text{in}} \left( \mathcal{C}^{(l')\dagger}(O) - \mathcal{C}^{\dagger}(O) \right) \right] \right| \le \epsilon
\end{equation}
holds with probability at least $1-\delta$. The sufficient truncation weight $l'$ is given by
\begin{equation}
    l'
    =
    \mathcal O\!\left(
    \frac{\log(1/(\epsilon\delta))}
    {\log\left(3/[2(1-\gamma)^2]\right)}
    \right).
\end{equation}
     
\end{lemma}
The learning algorithm is similar to the constant noise situation, which demonstrated in Algorithm~\ref{alg:qpl}. While the constant-noise analysis reveals the intrinsic sparsity $M e^{\mathcal O(l')}$, it relies on the path-damping mechanism and the associated logarithmic-depth regime. This Theorem that applies to arbitrary $\gamma\in[0,1)$ and arbitrary depth. Instead of identifying the circuit-dependent legal-path set, the algorithm learns all low-weight terminal Pauli coefficients in $\mathcal T_{l'}$, whose size is \(n^{\mathcal O(l')}\). With \(l'=\mathcal O(\log(1/\epsilon))\), this gives a general quasi-polynomial learning guarantee, at the cost of replacing the refined \(M e^{\mathcal O(l')}\) count by the universal enumeration \(n^{\mathcal O(l')}\). 

Furthermore, if the circuit has a $D$-dimensional local geometry and the observable $O$ has constant-size initial support, the relevant terminal strings lie within the light cone of $O$, yielding $ |\mathcal T_{l'}^{\rm LC}| \le \sum_{r=0}^{l'} 3^r \binom{M_{\rm LC}}{r} \le (O(d^D))^{l'}$, where $M_{\rm LC}$ is the number of qubits in the lightcone of $O$.
\begin{algorithm}
\caption{Quantum Process Learning Algorithm}
\label{alg:qpl}
\textbf{Input:} Data set $\mathcal{D}_{\rm QPT}=\left\{\ket{\psi_j}=\otimes^{n}_{i=1}\ket{\psi_{i,j}}\right\}_{j=1}^{N_{\rm data}}$ and accuracy parameter $\epsilon$;

\textbf{Output:} A $f(\cdot)$ such that $\abs{f(\cdot)-{\rm Tr}\left[O\mathcal{C}(\cdot)\right]}\leq\epsilon$ with high success probability for all input quantum states;

Let $l^{\prime}=[\log(1/\epsilon)]$, enumerate all the $P\in \mathcal{P}_n$ with $|P|\leq l^{\prime}$;

\textbf{For} $j\in[N_{\rm data}]$:

\quad Take the input state $|\psi_j\rangle\langle\psi_j|$ into the target quantum process, and obtain the output $\phi_j={\rm Tr}\left[O\mathcal{C}(|\psi_j\rangle\langle\psi_j|)\right]$;


\textbf{End For}

\textbf{For} each $P\in \mathcal{P}_n$ with $|P|\leq l^{\prime}$:

\quad Compute $\beta_{P}   =\frac{3^{|P|}}{N_{\rm data}}\sum_{j=1}^{N_{\rm data}}\phi_j\bra{\psi_j}P\ket{\psi_j}$.

\textbf{End For}


\textbf{Output}: $f(\cdot) = \sum_{|P|\leq l^{\prime}}\beta_{P}{\rm Tr}(P(\cdot))$
\end{algorithm}
\begin{table*}[t]
\caption{ Comparisons of our results with related previous studies on solving the quantum process learning problem, focusing on input distribution, target channel, sample complexity and classical post-processing complexity}

\renewcommand{\arraystretch}{1.6}
\label{tab:comparison}

\centering{
\begin{tabular}{l|cccc}
\hline\hline
\multicolumn{1}{c|}{\bf Algorithm} & \multicolumn{1}{c}{\bf Input Distribution} & \multicolumn{1}{c}{\bf Channel}  & \multicolumn{1}{c}{\bf Sample Complexity ($N_{\rm data}$)} & \multicolumn{1}{c}{\bf Classical Runtime}\\ \hline

Huang et al.~\cite{huangLearningPredictArbitrary2023} & \begin{tabular}[c]{@{}c@{}}Locally Flat Distribution\end{tabular} & CPTP  & $2^{\mathcal{O}(\log(n) \log(1/\epsilon))}$ & $\mathcal{O}(N_{\rm data})$\\ \hline
Chen et al.~\cite{chenPredictingQuantumChannels2024} &  \begin{tabular}[c]{@{}c@{}}Product State \end{tabular}& CPTP  & $\min\left(\frac{2^{\mathcal{O}(n)}}{\epsilon^2},n^{\mathcal{O}\left(\log\epsilon^{-1}/\log\frac{1}{1-\eta}\right)}\right)\cdot\log\frac{1}{\delta}$ & $\mathcal{O}(N_{\rm data})$\\ 

\hline
Raza et al.~\cite{razaOnlineLearningQuantum2024} & No restriction & Pauli Channel &\begin{tabular}[c]{@{}c@{}} $\mathcal{O}\left(\sqrt{n}\log(M)(\log{\frac{1}{\epsilon\delta})^{ 3/2}})\epsilon^{-3}\right)$\end{tabular} & $\mathcal{O}(N_{\rm data} \cdot 4^n)$ \\ 
\hline
Crupi et al.~\cite{crupiEfficientCharacterizationCoherent2025} & No restriction & $w$-low-degree noise
&$e^{\mathcal{O}(w\log(1/\epsilon))}\cdot
\mathcal{O}\!\left(n\log\frac{1}{\delta}\right)$
& $\mathcal{O}(n\cdot N_{\rm data})$
 \\
\hline
\textbf{Our Result 1} & No restriction  & \begin{tabular}[c]{@{}c@{}}$\mathcal{O}(\log n)$ depth RQC with \\ constant noise\end{tabular}  & $\mathcal{O}\left( 6^{l'_1} \cdot \epsilon^{-2} \log(1/\delta) \right)$ & $\mathcal{O}(n \cdot N_{\rm data}\cdot 2^{\mathcal{O}(l'_1)})$\\
\hline
\textbf{Our Result 2} & No restriction  & \begin{tabular}[c]{@{}c@{}}RQC with \\ arbitrary noise\end{tabular}  & $\mathcal{O}\left( n^{l'_2} \cdot \epsilon^{-2} \log(1/\delta) \right)$ & $\mathcal{O}(n^{\mathcal{O}(l'_2)} \cdot N_{\rm data})$\\\hline\hline
\end{tabular}
}

\begin{flushleft}
\scriptsize
\textbf{Note:} $\eta \in(0,1)$ relates to input distribution; $M$ is the number of Pauli terms within observable $O$. $w$ denotes the low-degree order: the noise Kraus operators act nontrivially on at most $w$ qubits at a time in their Pauli decomposition. $l'_1 = \mathcal{O}\left(\log(1/ \epsilon \delta)\right)$  and $ l'_2 = \mathcal O\!\left(
    \frac{\log(1/(\epsilon\delta))}
    {\log\left(3/[2(1-\gamma)^2]\right)}
    \right)$is the Pauli weight truncation threshold, representing the dominant term in the exponent.
\end{flushleft}

\end{table*}

While Refs.~\cite{huangLearningPredictArbitrary2023, chenPredictingQuantumChannels2024} utilized specific input state ensembles and Refs.~\cite{shaoSimulatingNoisyVariational2024,aharonovPolynomialTimeClassicalAlgorithm2023} relied on constant-level noise channels to induce this decay, we demonstrate that this phenomenon is naturally inherent in locally random circuits, across a noise strength $\gamma$ ranging from $0$ to a constant, and it is independent of the quantum circuit depth. 
While simulation algorithms must track the entire path history\citep{aharonovPolynomialTimeClassicalAlgorithm2023,angrisaniClassicallyEstimatingObservables2025}, our approach leverages the fact that all preceding evolution and noise effects are implicitly contracted within the learned coefficients $\{\beta_P\}$, allowing the complexity to be governed solely by the terminal Pauli weight $l'$ which remains independent of the circuit depth $d$. We summarize our third main result in the following Theorem. The whole proof is given in Appendix~\ref{sec:pt2}.

\begin{theorem}[Arbitrary Noisy Quantum Process Learning]
For any noisy quantum process $\Ccal$ given by Eq.~\ref{Equ:noisychannel}, and an $n$-qubit observable $O=\sum_{k=1}^M c_k Q_k$ with $M={\rm poly}(n)$, there exists a learning algorithm that learns a function $f$ from measurement results, satisfying $\abs{f(\rho_{\text{in}}) - {\rm Tr}[O \mathcal{C}(\rho_{\text{in}})]} \le \epsilon$ with success probability $\geq 1-\delta$. The sample complexity and the classical post-processing complexity are $n^{\mathcal{O}(\log \epsilon^{-1})}$ for any noise strength $\gamma$.

Specifically, when the circuit has a $D$-dimensional local geometry, the sample complexity is
$(d^D)^{\mathcal O(\log\epsilon^{-1})}$ and the classical post-processing complexity is $n\cdot(d^D)^{\mathcal O(\log\epsilon^{-1})}$.

\label{the:pl}
\end{theorem}

To contextualize our contribution, we compare the complexity and applicability of our protocol with state-of-the-art quantum learning algorithms in Table~\ref{tab:comparison}. A central theoretical contribution of our framework is the unified treatment of both i.i.d. single-qubit unital and non-unital noise channels. This stands in sharp contrast to complexity results derived for the $n$-qubit Pauli channel model \cite{razaOnlineLearningQuantum2024}. While their framework captures global correlations, it remains inherently restricted to unital dynamics, whereas our approach tackles the distinct challenge of generic non-unitality. The capacity of our framework to handle non-unital CPTP maps is essential for modeling realistic hardware decoherence, such as amplitude damping, which remains inaccurate under unital approximations.

Complementing this noise universality, our QPT algorithm yields a strictly \emph{input-agnostic} characterization. Unlike protocols that rely on input distributions that are at most polynomially far from locally flat distributions \cite{huangLearningPredictArbitrary2023} or specific product-state ensembles \cite{chenPredictingQuantumChannels2024}, our approach operates independently of the input distribution. This guarantees predictive robustness for arbitrary input states, effectively bypassing the distribution-shift limitations inherent in previous works. Our framework remains valid even for input states that violate standard distribution assumptions, such as highly entangled states (see Appendix~\ref{sec:highly entangled} for  preliminary numerical tests on entangled inputs). Another closely related approach is the low-degree noise tomography protocol of Ref.~\cite{crupiEfficientCharacterizationCoherent2025}. While both approaches exploit a reduced Pauli representation, the origin of the reduction is different: their efficiency relies on a fixed low-degree noise ansatz, whereas our truncation is justified by the decay property induced by random circuits with i.i.d. single-qubit noise. 

For constant accuracy $\epsilon$, the constant-noise guarantee in Our Result 1 has sample complexity proportional to $6^{l'_1}$, where $l'_1=\mathcal{O}(\log(1/\epsilon\delta))$ and is independent of $n$. Unlike the fixed low-degree tomography of Ref.~\cite{crupiEfficientCharacterizationCoherent2025}, whose gate-layer extension incurs an additional linear-in-$n$ overhead from classical post-processing, our effective truncation threshold follows from the decay property established for noisy random circuits. This also contrasts sharply with the method of Ref.~\cite{razaOnlineLearningQuantum2024}, which suffers from a $\mathcal{O}(\sqrt{n})$ polynomial dependency on the system size $n$, and an intractable $\mathcal{O}(4^n)$ classical runtime bottleneck inherent to standard shadow tomography protocols when estimating general channels.


When considering higher precision, such as $\epsilon=1/n$, the constant-noise result remains polynomial in $n$ in the logarithmic-depth regime because $l'_1=\mathcal{O}(\log n)$. For arbitrary noise strength, Our Result 2 instead uses the universal terminal enumeration with complexity $n^{l'_2}$; thus it gives a depth-independent and input-agnostic guarantee, but its scaling becomes quasi-polynomial when $l'_2=\mathcal{O}(\log n)$. This separates the two conclusions in Table~\ref{tab:comparison}: the constant-noise regime yields the sharper polynomial scaling, whereas the arbitrary-noise regime provides the more general noise-strength-agnostic guarantee.
\vspace{6px}

\vspace{8px}
\section{Numerical Results}
\subsection{Performance on the Ising Model}
\begin{figure*}[t]
\begin{center}
\includegraphics[width=1\linewidth]{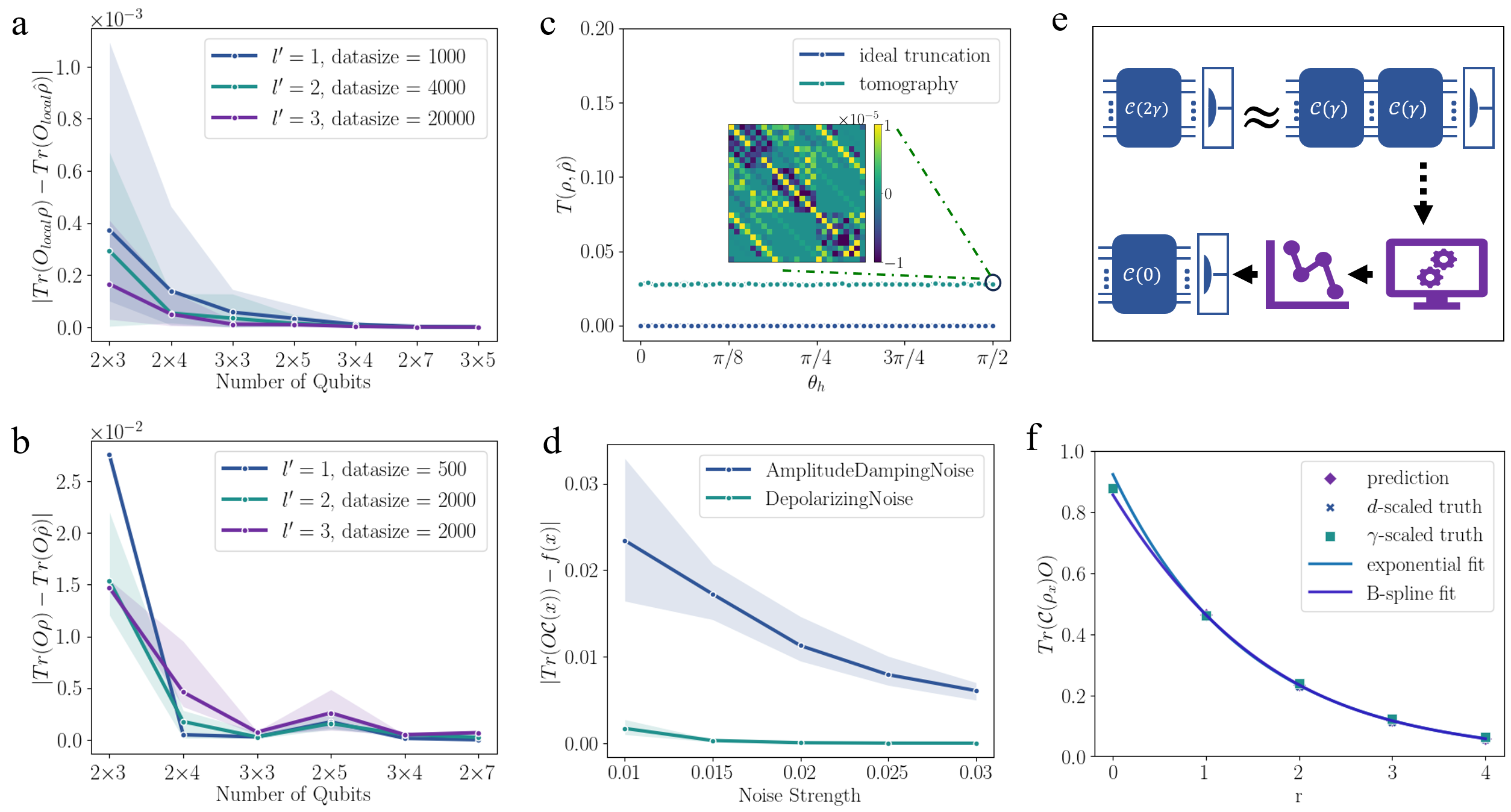}
\end{center}
\caption{(a) QST results for various numbers of qubits and $l'$. Each circuit is 20 layers accompanied by depolarizing noise of strength 0.02 and fixed $\theta_h=\frac{\pi}{4}$. The grid illustrates the $3\times5$ $2D$ transverse field Ising model. (b). QPT results for different qubit numbers and $l'$, where the circuit depth is 5 layers and the depolarizing noise strength is 0.01. (c). Learning of the $\rho$ generated by sweeping $\theta_h$ from $0$ to $\frac{\pi}{2}$, where the trace distance $T(\rho,\hat{\rho})=\frac{1}{2}\|\rho-\hat{\rho}\|_1$; the circuit size $2\times 5$, $45$ layers, depolarizing noise strength 0.02. The heat-map shows a $25\times25$ sub-matrix of the matrix $\langle i|\left(\rho-\hat{\rho}\right)|j\rangle$ at $\theta_h=\frac{\pi}{2}$, where basis $i,j\in\{0,1\}^n$ (the full matrix in Appendix~\ref{sec:expr}).  (d). QPT for the $2\times5$ system under 2 kinds of noise and other settings identical to b. (e).The procedure of the ZNE. (f) The numerical result of the ZNE-QEM using the proposed learning algorithm.} 
\label{fig:exp}
\end{figure*}
\noindent
Although our theoretical results rely on the randomness assumption, we numerically verify that our learning algorithm remains highly efficient for a broader class of circuits, including those with specific structure, such as noisy quantum dynamical processes. This demonstrates the broad practical applicability of our approach. 

Specifically, we benchmark the noisy Hamiltonian dynamics of 2D transverse-field Ising model given by $H = -J \sum_{\langle q,p \rangle}Z_q Z_p + h \sum_q X_q$. The time evolution is implemented via Trotterization as a sequence of ${\rm R_{ZZ}}(\theta_J=-\pi/2)$ and ${\rm R_{X}}(\theta_h)$ gates~\citep{kimEvidenceUtilityQuantum2023}. Using Qulacs~\citep{2021suzukiQulacsFastVersatile}, we simulate noisy circuits under two conditions: i.i.d.\ single-qubit depolarizing noise ($p=0.02$) for the unital case, and a mixture of depolarizing and amplitude damping for the non-unital case. Our simulations cover system sizes up to $3\times 5$ qubits and 20 layers, a scale that occupies $16$~GB RAM.

The numerical results are presented in Fig.~\ref{fig:exp}, where shaded areas indicate the outcome range over 10 trials. Consistent with Theorems~\ref{the:sl} and~\ref{the:pl}, Figs.~\ref{fig:exp} (a) and (b) show that the reconstruction error stabilizes at large system sizes, confirming the localized nature of the learning problem. Figs.~\ref{fig:exp} (c) and (d) further demonstrate robustness against parameter variations and non-unital noise; interestingly, stronger noise yields higher accuracy, reflecting the tighter Pauli weight confinement. Beyond sample efficiency, our method enables extreme memory compression: a 14-qubit state description is reduced from $8$~GB (full matrix) to just $1$~KB (sparse coefficients). 

\subsection{Application: Quantum Error Mitigation}
\label{sec:apply}

We note that our learning algorithm can also be applied to solve the quantum error mitigation (QEM) task~\citep{eisertQuantumCertificationBenchmarking2020}. QEM comprises protocols that suppress stochastic errors on NISQ hardware by classical post-processing of measurement data, without introducing full quantum error-correcting codes. Whereas error correction aims to eliminate noise, QEM converts every hardware improvement into an immediate fidelity gain by suppressing residual errors. One QEM approach is zero-noise extrapolation (ZNE), which executes the circuit at several circuit fault rates $\lambda$, which measures the level of errors occurring in the overall circuit, and $ \lambda\propto \gamma$ \citep{caiQuantumErrorMitigation2023}. Although the circuit output at $\lambda=0$ cannot be measured directly, an empirical model $h(\lambda)$ linking  $\lambda$ to the circuit output can be built from a set of different $\lambda$ values. This allows us to extrapolate the case of $\lambda=0$, which corresponds to zero noise.  Different $\lambda$ values can be generated by pulse-stretching\citep{kandalaErrorMitigationExtends2019, kimScalableErrorMitigation2023} or by inserting additional noise channels \citep{endoPracticalQuantumError2018}. For an i.i.d single-qubit noise,  it is natural to set $\lambda$ proportional to the gate count, thus $\gamma \propto \lambda\propto \text{Number\ of\ the\ gates}$. Here we vary $\lambda$ by controlling the depth of the circuit $d$.

Whereas the conventional ZNE must be tailored to each specific input, our protocol is input-agnostic. Similar to Lemma~\ref{le:noise_main}, ${\rm{Tr}}(O\mathcal{C}(\cdot)) = \sum_{|P|\leq l'}(1-\gamma)^{|P|}\Phi(\mathcal{C},P){\rm{Tr}}(P\cdot)$. Considering
the depolarizing noise strength $\gamma<1$, $(1-\gamma)^{|P|} =(1-\gamma)^{\frac{|P|}{d}d}\approx(1-\gamma d)^{\frac{|P|}{d}} $. In other words, one can obtain the characterizations of the same quantum processes with different noise strength by appending extra quantum circuit layers to the original process, this yields a sequence of learned values  $\{f_r| \abs{f_r-{\rm Tr}[O\mathcal{C}_{rd}(|0^n\rangle\langle0^n|)]}\leq \epsilon\}_{r\geq1}$. One can utilize $\{f_r\}_{r\geq1}$ to extrapolate $f_0$, which is considered as the characterization of $\mathcal{C}_d$ with zero noise.

The result of numerical experiments of application is shown in Fig~\ref{fig:exp}(f), where we simulate a six-qubit 2D transverse field Ising model 
\begin{equation}
    \label{equ:sim}
    H = -J \sum_{\langle q,p \rangle}Z_q Z_p + h \sum_q X_q,
\end{equation}
with $5$ layers.

Two key observations emerge:

\begin{itemize}
    \item Rescaling either the depth coefficient $d$ or the noise strength $\gamma$ perturbs the dynamics to a comparable extent, as seen from the nearly overlapping dots.
    \item The characterization obtained by learning the coefficients of $d$ can be extrapolated via curve fitting to estimate the noise-free system (i.e., when $\gamma=0$) characterization. Exponential extrapolation yields an error $0.0446$;  a cubic B-spline (piecewise polynomial) reduces it to $0.0222$.
\end{itemize}
\vspace{8px}
\section{Discussion}

\noindent
Efficiently characterizing noisy quantum states and processes is a fundamental challenge. In this work, we develop a learning framework for noisy random quantum circuits that covers both noisy quantum state tomography and noisy quantum process tomography. For unital noisy states, Pauli-path damping yields a compact learnable representation in the constant-noise, logarithmic-depth regime, while deeper constant-noise circuits are already close to the noise-induced fixed point at the target precision. For process learning, the Heisenberg-picture formulation extends the framework to both unital and non-unital i.i.d. single-qubit noise channels and enables prediction for arbitrary input states. The process-learning guarantees split into two complementary regimes: constant noise gives a sharper polynomial scaling in the logarithmic-depth regime, whereas arbitrary noise strength admits a depth-independent, input-agnostic algorithm with polynomial scaling in the constant accuracy, maintaining sample and runtime efficiency across the entire noise spectrum, from the noiseless limit ($\gamma=0$) to constant noise levels ($\gamma=\Theta(1)$)

Functionally, our work serves as a ``learning-theoretic dual'' to classical simulation~\citep{gil-fusterRelationTrainabilityDequantization2025}. While simulation algorithms predict outcomes by aggregating dominant Pauli paths based on known circuit parameters~\citep{aharonovPolynomialTimeClassicalAlgorithm2023,angrisaniSimulatingQuantumCircuits2025,schusterPolynomialtimeClassicalAlgorithm2025}, our approach identifies the effective terminal Pauli coefficients directly from measurement data, solving the inverse problem without requiring prior knowledge on circuit depth or architecture~\citep{chen2023learnability}. This offers a complementary benchmarking perspective to standard certification methods~\citep{eisertQuantumCertificationBenchmarking2020}. 


Several open problems remain for future study. First, while our main theorems establish average-case efficiency, we also derive a sample complexity lower bound for the worst-case scenario (see Appendix~\ref{Sec:samplelowerbound}), underscoring the inherent hardness of QST tasks in such cases. This raises the question: can we theoretically characterize the efficiency boundaries of QST and QPT tasks? Furthermore, we show that sample complexity of QPT task is a smooth function of the noise strength $\gamma$, exhibiting a fundamental difference from the computational hardness results in random circuit sampling, which feature a sudden complexity transition. Exploring the formal relationship between learning complexity and computational complexity remains an intriguing direction. Finally, extending our i.i.d. analysis to gate-dependent noise models is crucial for characterizing more practical quantum hardware.

\section{Acknowledgments}
This work is supported by National Natural Science Foundation of China (NSFC) Grants (No.~62501060, No.~62461160263, No.~12361161602 and  No.~62131002), NSAF (Grant No.~U2330201), Beijing Natural Science Foundation Z250004, 
Beijing Science and Technology Planning Project (Grant No.~Z25110100810000)
Quantum Science and Technology-National Science and Technology Major Project (2023ZD0300200), and the High-performance Computing Platform of Peking University and Pawsey Supercomputing Research Centre.

\clearpage
 \bibliography{main}
 \clearpage
\widetext
\section*{Appendix}
\appendix

\section{Structure and Applicability of the noisy circuit in this work}
\label{sec:introC}
We emphasize that the quantum process studied serves as a standard model with wide and practical applications, especially in the Near-Term Intermediate Scale Quantum (NISQ) era. This appendix details the topological definitions, generality, and practical relevance of the quantum circuit model investigated.

\subsection{Topological Definition and Model Generality}

The studied noisy quantum process $\mathcal{C}$ adopts a layered structure, representing a large class of quantum circuits:
\begin{equation}
    \begin{aligned}
        \mathcal{C}=\mathcal{E}^{\otimes n}\mathcal{C}_d\mathcal{E}^{\otimes n}\mathcal{C}_{d-1}\cdots\mathcal{E}^{\otimes n}\mathcal{C}_1,
    \end{aligned}
\label{Eq:noisychannel}
\end{equation}

in which a $\gamma$-strength local noise channel $\mathcal{E}$ (unital or non-unital) is applied uniformly throughout the circuit. The quantum circuit depth is $d$, and each layer $\mathcal{C}_i$ comprises non-overlapping two-qubit gates acting on arbitrary pairs, where each gate is uniformly sampled from a local 2-design unitary group. Each layer $\mathcal C_{i}$ can realise variational algorithms such as VQE, QAOA or quantum neural networks \citep{zhouQuantumApproximateOptimization2020,kandalaHardwareefficientVariationalQuantum2017, grimsleyAdaptiveVariationalAlgorithm2019}, while the local noise assumption $\mathcal E^{\otimes n}$ captures typical NISQ imperfections.
The local 2-design assumption is an extremely weak condition, where quantum neural network models are typical cases\citep{mccleanBarrenPlateausQuantum2018,cerezoCostFunctionDependent2021}, and even Clifford gates satisfy such an assumption \citep{zhuCliffordGroupFails2016}. We note that if an ensemble follows a $(t+1)$-design, it must follow the $t$-design property \citep{meleIntroductionHaarMeasure2024}. As a result, this assumption is very general and covers a large amount of NISQ algorithms related to 'randomly initial parameters' and 'classical optimizations'\citep{mccleanBarrenPlateausQuantum2018,cerezoCostFunctionDependent2021}.

The circuit model is formally defined below using graph-theoretic definitions:

\begin{definition}[Architecture, restatement of Ref.~\cite{haferkampLinearGrowthQuantum2022}]
\label{def:arch}
\emph{An architecture is a directed acyclic graph that contains $R\in \mathbb{Z}_{>0}$ vertices (gates). Two edges (qubits) enter each vertex, and two edges exit. Two typical examples are listed below:}
\begin{itemize}
\item \emph{A brickwork is the architecture of any circuit formed as follows. Apply a string of two-qubit gates: $U_{1, 2}\otimes U_{3, 4}\otimes \cdots \otimes U_{n-1,n}. $Then apply a staggered string of gates. Perform this pair of steps $T$ times in total, using possibly different gates each time.} 
\item \emph{A staircase is the architecture of any circuit which applies a stepwise string of two-qubit gates: $U_{n,n-1}U_{n-2,n-1}\cdots U_{2, 1}$. Repeat this process $T$ times, using possibly different gates each time.}
\end{itemize}
\end{definition}
Here, the quantum circuit layer $\mathcal{C}_i$ may adopt any architecture, and \emph{\textbf{we note that our learning algorithm can be applied to any geometrical architecture,}} and thus covers a large class of noisy quantum circuits, especially for those used in NISQ algorithms.

\begin{definition}[Random Quantum Circuit, restatement of Ref.\cite{haferkampLinearGrowthQuantum2022}]
\label{RQC}
\emph{
Let $G$ denote an arbitrary architecture. A probability distribution can be induced over the architecture-$G$ circuits as follows: for each vertex in $G$, draw a gate Haar-randomly from $SU(4)$. Then contract the unitaries along the edges of $G$. Each circuit so constructed is called a random quantum circuit.
}
\end{definition}

\begin{definition}[Random noisy quantum circuit]
\label{RQC_noisy}
\emph{
Let $\widetilde G$ denote an arbitrary architecture. A probability distribution can be induced over the architecture-$\widetilde G$ circuits as follows: for each vertex in $\widetilde G$, draw a gate Haar-randomly from $SU(4)$ and an i.i.d single-qubit noisy channel. Then contract the unitaries along the edges of $\widetilde G$. Each circuit so constructed is called a random noisy quantum circuit.}
\end{definition}
\noindent
The guarantee is a high-probability bound ($\geq 1-\delta$) over random circuit ensemble defined in Definition.~\ref{RQC_noisy}. Furthermore, we numerically demonstrate that our learning algorithm can successfully handle a noisy Hamiltonian dynamics approach, where the underlying quantum circuit does not possess the locally random property. 

\subsection{Importance for Quantum Benchmarking and Learning}

To design powerful quantum algorithms, such as quantum neural network models and related states, a benchmarking algorithm is necessary \citep{arute2019quantum,babbushGrandChallengeQuantum2025}; otherwise, one may not verify and check the correctness of the implemented quantum algorithm. Following this logic, a large amount of quantum learning algorithms are proposed for quantum state (process) tomography, Hamiltonian learning\citep{haahLearningQuantumHamiltonians2024} , shallow circuit learning\citep{huangLearningShallowQuantum2024}, quantum gate tomography, and other quantum benchmarking algorithms. \emph{\textbf{To the best of our knowledge, this is the first provably efficient learning algorithm for noisy state and process tomography}}, providing an efficient tool for verifying the output of the implemented quantum algorithms on NISQ devices.

\subsection{The Gate-Independent Noise Model}

We utilize the gate-independent noise model, which posits that the detrimental effects impacting quantum operations are uniform across all fundamental gates, irrespective of their specific type or physical implementation. This simplifying assumption is widely adopted due to several key factors:

\begin{itemize}
\item \textbf{Theoretical Tractability:} Adopting a gate-independent noise assumption allows researchers to advance the development and analysis of error correction protocols and fault-tolerant methodologies without needing to incorporate the intricate details of gate-specific noise characteristics\citep{knill2008randomized,helsen2019new,chen2021robust}. This uniformity facilitates the derivation of universal results and theoretical performance bounds \citep{nielsen2001quantum}.
\item \textbf{Practical Approximations:} In particular quantum systems---especially those featuring highly calibrated gates acting on the same number of qubits and employing standardized control mechanisms---the variability of noise across different gates can be negligible\citep{shor1996fault,arute2019quantum}. In these instances, the gate-independent noise model serves as a tenable approximation, streamlining analysis without substantially compromising precision.
\item \textbf{Alignment with Noise Conversion Methods (Twirling):} Techniques like Pauli twirling are routinely applied to convert complicated physical noise channels into simpler, diagonal forms in the Pauli basis\citep{wallman2016noise,chen2023learnability}. The resulting channel can often be effectively approximated as gate-independent, thereby conforming to the model's postulates.
\end{itemize}
The gate-independent noise model thus furnishes a foundational framework for comprehending error propagation and engineering correction strategies. We identify the robust depiction of gate-dependent noise, which typically manifests in larger, more intricate quantum architectures, as a significant avenue for future exploration.

\section{Learning a Quantum state}
\label{sec:pauli}
The representation of noisy quantum state are presented in this section, together with further implementation details of the QST algorithm.

\subsection{ Noisy Quantum State}
In this section, we will give an unified representation of noisy quantum state under Pauli propagation.
\begin{lemma}
[Unified Representation of Noisy Quantum State]
\label{le:noise appendix}
Let the noisy quantum state $\rho=\mathcal{C}(|0^n\rangle\langle0^n|)$ with $\mathcal{C}=\mathcal{E}^{\otimes n}\mathcal{C}_d\mathcal{E}^{\otimes n}\mathcal{C}_{d-1}\cdots\mathcal{E}^{\otimes n}\mathcal{C}_1$ representing a $d$-depth noisy quantum circuit, where $\mathcal{C}_i(\cdot)=C_i(\cdot)C_i^{\dagger}$ is a unitary channel consisting of a layer of two-qubit gates, and $\mathcal{E}$ is a unital single-qubit noise channel with strength parameter $\gamma$. Then the noisy quantum state $\rho$ can be represented by the Pauli path integral, that is 
\begin{align}
    \rho=\sum_{s \in \tilde{\mathcal{P}}_n^{\otimes (d+1)}}(1-\gamma)^{|s|}\Phi(\mathcal{C},s)s_d,
\end{align}
where the $n(d+1)$-qubit operator $s=s_0s_1\cdots s_d$, $\tilde{\mathcal{P}}_n=\{I/\sqrt{2},X/\sqrt{2},Y/\sqrt{2},Z/\sqrt{2}\}^{\otimes n}$. The Pauli weight $|s|$ represents the number of non-identity operators in $s\in\tilde{\mathcal{P}}_n^{\otimes (d+1)}$. The coefficient
\begin{equation}
\Phi(\mathcal{C},s)={\rm{Tr}}(s_d\mathcal{C}_d(s_{d-1}))
\cdots{\rm{Tr}}(s_1\mathcal{C}_1(s_0))\langle 0^n|s_0|0^n\rangle.
\end{equation}
\end{lemma}
We prove it by describing  types of noisy channels, which are depolarizing noise, single-qubit Pauli noise, and non-unital noise.
\label{sec:plen}
\subsubsection{Depolarizing Noise}
The property of depolarizing noise $\mathcal{E}_{{\rm depo}}$ is that
\begin{equation}
    \begin{aligned}
        \mathcal{E}_{{\rm depo}}(I) &= I,\\
        \mathcal{E}_{{\rm depo}}(X) &= (1-\gamma)X,\\
        \mathcal{E}_{{\rm depo}}(Y) &= (1-\gamma)Y,\\
        \mathcal{E}_{{\rm depo}}(Z) &= (1-\gamma)Z,\\
    \end{aligned}
\end{equation}
so that
\begin{equation}
\begin{aligned}
    \rho &= \sum_{s\in{\tilde{P}_n^{d+1}}}s_d{\rm{Tr}}(s_d\mathcal{E}_{{\rm depo}}^{\otimes n}\mathcal{C}_d(s_{d-1}))
\cdots{\rm{Tr}}(s_1\mathcal{E}_{{\rm depo}}^{\otimes n}\mathcal{C}_1(s_0))\langle 0^n|s_0|0^n\rangle
    \\&=\sum_{s\in{\tilde{P}_n^{d+1}}}(1-\gamma)^{|s|}s_d\Phi(C,s),
\end{aligned}
\end{equation}
where $|s|$ means the number of non-identities in s. 
\subsubsection{Pauli Noise}
\label{sec:pnoise}
We use $\mathcal{E}$ to denote the noise function. Pauli noise $\mathcal{E}_{Pauli}$ is
\begin{equation}
    \mathcal{E}_{{\rm Pauli}}(\rho) =\gamma_1\rho+\gamma_2X\rho X^{\dagger}+\gamma_3Y\rho Y^{\dagger}+\gamma_4Z\rho Z^{\dagger},
\end{equation}
where $\gamma_1+\gamma_2+\gamma_3+\gamma_4 = 1$. The Pauli noise has the property that
\begin{equation}
\begin{aligned}
    \mathcal{E}_{{\rm Pauli}}(I) &= (\gamma_1+\gamma_2+\gamma_3+\gamma_4)I = I,\\
    \mathcal{E}_{{\rm Pauli}}(X) &= (\gamma_1+\gamma_2-\gamma_3-\gamma_4)X = (1-2(\gamma_3+\gamma_4))X,\\
    \mathcal{E}_{{\rm Pauli}}(Y) &= (\gamma_1-\gamma_2+\gamma_3-\gamma_4)Y = (1-2(\gamma_2+\gamma_4))Y,\\
    \mathcal{E}_{{\rm Pauli}}(Z) &= (\gamma_1-\gamma_2-\gamma_3+\gamma_4)Z = (1-2(\gamma_2+\gamma_3))Z.\\
\end{aligned}
\end{equation}
So the Pauli Channel can be written as
\begin{equation}
\begin{aligned}
    \rho &= \sum_{s\in{\tilde{\mathcal{P}}_n^{d+1}}}s_d{\rm{Tr}}(s_d\mathcal{E}_{{\rm Pauli}}^{\otimes n}\mathcal{C}_d(s_{d-1}))
\cdots{\rm{Tr}}(s_1\mathcal{E}_{{\rm Pauli}}^{\otimes n}\mathcal{C}_1(s_0)){\rm Tr}(s_0|0^n\rangle\langle 0^n|)
    \\&=\sum_{s\in{\tilde{\mathcal{P}}_n^{d+1}}}(1-2(\gamma_3+\gamma_4))^{|s|_X}(1-2(\gamma_2+\gamma_4))^{|s|_Y}(1-2(\gamma_2+\gamma_3))^{|s|_Z}s_d\Phi(C,s),
\end{aligned}
\end{equation}
where $|s|_{\tilde{P}}$ denotes the number of $\tilde{P}$ in $s$. Without loss of generality, to simplify the presentation while maintaining a conservative upper bound on path decay, we consider a uniform noise rate $\gamma = \min \{ \gamma_2+\gamma_3, \gamma_2+\gamma_4, \gamma_3+\gamma_4 \}$, where $0\leq\gamma<1$.  This substitution ensures that the decay factor for each Pauli component is upper-bounded by a uniform term $(1-\gamma)^{|s|}$, where $|s|$ is the total weight of the Pauli string. For brevity, we denote $\rho = \sum_{s \in \tilde{\mathcal{P}}_n^{d+1}} (1-\gamma)^{|s|} s_d \Phi(\mathcal{C}, s)$ in the subsequent discussion, which does not alter the analysis of the learning algorithm.

Techniques like Pauli twirling are employed to transform complex unital channels into diagonal forms on the Pauli basis~\citep{chen2023learnability,wallman2016noise}. In the following, we utilize the Pauli noise channel to represent the unital channel.
Thus, complete the proof of Lemma~\ref{le:noise appendix}.

\subsection{State Learning Truncation}
\label{sec:pnp}
Ref.~\cite{aharonovPolynomialTimeClassicalAlgorithm2023} proved that sampling from a depolarizing channel reduces to fitting a constant number $l$ of Pauli paths.
We generalize this observation to single-qubit Pauli noise that admits a sparse Pauli-path expansion.

According to Lemma~\ref{le:noise appendix}, for an arbitrary i.i.d single-qubit noise, the output state is approximated by
\begin{equation}
  \hat{\rho}=\sum_{|s_d|\leq l^{\prime}, s_d\in \tilde{\mathcal{P}}_n}\alpha_{s_{d}}s_d=\sum_{s\in{\tilde{\mathcal{P}}_n^{d+1}},|s|\leq l}(1-\gamma)^{|s|}s_d\Phi(\mathcal{C},s).
\end{equation}
In other words, it is sufficient to learn finitely many low-weight legal Pauli paths in order to obtain an approximation $\hat{\rho}$ satisfying $\|\rho-\hat{\rho}\|_2<\epsilon_1$, where $\|A\|_2$ denotes the standard Frobenius norm. The formal statement and proof are given in Lemma~\ref{le:pna}.
\begin{lemma}[The unital noisy state truncation]
\label{le:pna}
Let the noisy quantum state $\rho=\mathcal{C}(|0^n\rangle\langle0^n|)$ with $\mathcal{C}=\mathcal{E}^{\otimes n}\mathcal{C}_d\mathcal{E}^{\otimes n}\mathcal{C}_{d-1}\cdots\mathcal{E}^{\otimes n}\mathcal{C}_1$ representing a $d$-depth noisy quantum circuit, where $\mathcal{C}_i$ is a layer of two-qubit Haar random quantum gates. Suppose that the single-qubit noise channel is unital with constant noise strength.  Then there exists a classical representation $\hat{\rho}=\sum_{|P|\leq l^{\prime}, P\in \mathcal{P}_n}\alpha_P P$ such that
\begin{equation}
    ||\rho-\hat{\rho}||_2<\epsilon_1,
\end{equation}
where coefficients $\alpha_{s_{d}}\in\mathbb{R}$ and $l^{\prime}=\mathcal{O}\left( \log(\epsilon_1^{-1}\delta_1^{-1})\right)$ with success probability at least $1-\delta_1$. 
\end{lemma}

\begin{proof}
We first consider the Frobenius norm. Let
\begin{equation} 
\begin{aligned}
\label{equ:delta}
  \Delta &:=\|\rho-\hat{\rho}\|_2
  \\&=\sqrt{{\rm Tr}\left(\left(\sum_{|s|>l}(1-\gamma)^{|s|}s_d\Phi(\mathcal{C},s)\right)\left(\sum_{|s|>l}(1-\gamma)^{|s|}s_d\Phi(\mathcal{C},s)\right)^\dagger\right)}
  \\&= \sqrt{{\rm Tr}\left(\sum_{|s|>l}\sum_{|s'|>l}(1-\gamma)^{|s'|+|s|}s_ds_d^{\prime\dagger}\Phi(\mathcal{C},s)\Phi(\mathcal{C},s')\right)}
  \\& = \sqrt{\sum_{|s|>l}\sum_{|s'|>l}(1-\gamma)^{|s'|+|s|}\Phi(\mathcal{C},s)\Phi(\mathcal{C},s'){\rm Tr}\left(s_d s_d^{\prime\dagger}\right)}.
\end{aligned}
\end{equation}
This expression reduces the truncation error to the second moments of the Pauli-path amplitudes.

For distinct Pauli paths, the local random circuit ensemble gives the second-moment orthogonality
\begin{equation}
\mathbb{E}_{\mathcal{C}}[\Phi(\mathcal{C},s)\Phi(\mathcal{C},s')]=0.
\label{Eq:random orthogonality}
\end{equation}
Therefore,
 \begin{equation}
\begin{aligned}
\mathbb{E}_{\mathcal{C}}(\Delta^2)&=\mathbb{E}_{\mathcal{C}}\left(\sum_{|s|>l}\sum_{|s'|>l}(1-\gamma)^{|s'|+|s|}\Phi(\mathcal{C},s)\Phi(\mathcal{C},s'){\rm Tr}\left(s_ds_d^{\prime\dagger}\right)\right)
\\&=\mathbb{E}_{\mathcal{C}}\left(\sum_{|s|>l}(1-\gamma)^{2|s|}\Phi(\mathcal{C},s)^2{\rm Tr}\left(s_ds_d^{\dagger}\right)\right)
\\&=\mathbb{E}_{\mathcal{C}}\left(\sum_{|s|>l}(1-\gamma)^{2|s|}\Phi(\mathcal{C},s)^2\right)
\\& = \sum_{k>l}(1-\gamma)^{2k}W_k.
\end{aligned}
\end{equation}
The second line follows from the orthogonality above, and the third line uses ${\rm Tr}(s_ds_d^{\dagger})=1$ for normalized Pauli operators. In the last line we define $W_k = \mathbb{E}_{\mathcal{C}}\sum_{|s|=k}\Phi(\mathcal{C},s)^2$.

We now show that the quantities $W_k$ form a normalized nonnegative weight distribution. Consider the corresponding noiseless Pauli-path expansion
\begin{equation}
    \rho_0:=\mathcal{C}_d\mathcal{C}_{d-1}\cdots\mathcal{C}_1(|0^n\rangle\langle0^n|)
    =\sum_{s\in\tilde{\mathcal{P}}_n^{\otimes(d+1)}}\Phi(\mathcal{C},s)s_d .
\end{equation}
For an arbitrary circuit, $\rho_0$ is a density state and therefore $\|\rho_0\|_2^2={\rm Tr}(\rho_0^2)\leq 1$. Taking the ensemble average and expanding the right-hand side gives
\begin{equation}
\begin{aligned}
\mathcal{O}(1)
&=\mathbb{E}_{\mathcal{C}}{\rm Tr}(\rho_0^2) \\
&=\sum_{s,s'}\mathbb{E}_{\mathcal{C}}\left[\Phi(\mathcal{C},s)\Phi(\mathcal{C},s')\right]
{\rm Tr}\left(s_d s_d^{\prime\dagger}\right).
\end{aligned}
\end{equation}
According to Eq~\ref{Eq:random orthogonality}, we obtain
\begin{equation}
    \mathcal{O}(1)=\mathbb{E}_{\mathcal{C}}\sum_s\Phi(\mathcal{C},s)^2
    =\sum_{k\geq0}\mathbb{E}_{\mathcal{C}}\sum_{|s|=k}\Phi(\mathcal{C},s)^2
    =\sum_{k\geq0}W_k .
\end{equation}

Thus, we have
\begin{equation}
\begin{aligned}
\mathbb{E}_{\mathcal{C}}(\Delta^2)
&=\sum_{k>l}(1-\gamma)^{2k}W_k \\
&\leq (1-\gamma)^{2l}\sum_{k>l}W_k \\
&\leq (1-\gamma)^{2l}\sum_{k\geq0}W_k \\
&\leq (1-\gamma)^{2l}.
\end{aligned}
\end{equation}
By Markov's inequality,
\begin{equation}
    \mathbb{P}\left(\|\rho-\hat{\rho}\|_2\geq\epsilon_1\right)
    \leq\frac{\mathbb{E}_{\mathcal{C}}(\Delta^2)}{\epsilon_1^2}.
\end{equation}
Therefore, it is sufficient to choose $l=\mathcal{O}\left(\frac{\log(\epsilon_1^{-1}\delta_1^{-1})}{\gamma}\right)$ such that
$\Delta\leq\epsilon_1$.
Here, we assume that the noise strength is a constant independent of the system size, so that $l$ scales as $\mathcal{O}(\log(\epsilon_1^{-1}\delta_1^{-1}))$. Since $|s_d|<|s|$, $l'$ also scales as $\mathcal{O}(\log(\epsilon_1^{-1}\delta_1^{-1}))$.
With this choice, $\|\rho-\hat{\rho}\|_2\leq\epsilon_1$ holds with probability at least $1-\delta_1$.

\end{proof}

\subsubsection{Pauli Path Enumeration}
In this section, we demonstrate the theoretical lower bound of the number of Pauli paths with weight less than $l$.

We prove it by leveraging the legal path and the lower bound of the term at the form of the ${\rm Tr}(s_i\mathcal{C}_i(s_{i-1}))^2$.

\begin{definition}[legal Pauli path]
    A legal Pauli path is a Pauli path $s$ such that for each $i\in[d]$, ${\rm Tr}(s_i\mathcal{C}_i(s_{i-1}))^2\neq 0$.
\label{def:legal path}
\end{definition}
For unital noises, using the equation

\begin{equation}
\underset{U \sim \mathbb{SU}(4)}{\mathbb{E}}{\rm{Tr}}(xU (y))^{2}=\underset{U \sim \mathbb{SU}(4)}{\mathbb{E}}{\rm{Tr}}(xU yU^\dagger)^{2}
=\left\{
\begin{array}{ll}1, & x=y=I^{\otimes 2} / 2, \\ 0, & x=I^{\otimes 2} / 2, y \neq I^{\otimes 2} / 2, \\ 0, & x \neq I^{\otimes 2} / 2, y=I^{\otimes 2} / 2, \\ \frac{1}{15}, & \text { else. }\end{array}\right.
\end{equation}

We observe that certain Pauli paths contribute $ 0$ to the circuit; these are termed illegal Pauli paths. 

For $k=0$, $W_k=1$, where the Pauli path $s$ consists of identity operators. 

For $k\in(0,d]$, $W_k = 0$.

For $k\geq d+1$, we can bound $W_k$ by focusing on every term, which is in the form of $\mathbb{E}_{C_{i}}{\rm Tr}(s_i\mathcal{C}_i(s_{i-1}))^{2}$.Therefore, we have the following lemma.

\begin{lemma}[lower bound of legal Pauli paths]
    Considering a $d$-depth random quantum circuit $\mathcal{C}$, define $ L(\mathcal{C},s) = \prod_{i=1}^d {\rm Tr}(s_i\mathcal{C}_i(s_{i-1}))^2$ for a legal Pauli path $s$. Then, there exists a lower bound that $\mathbb{E}_{\mathcal{C}}[L(\mathcal{C},s)]\geq (1/15)^{|s|}$.
    \label{le:lower bound term}
\end{lemma}

\begin{proof}
Noting that each $C_i$ is a layer of two-qubit gates,  $C_i$ is equal to the multiplication of $ C_{i}^{(j)}$, where $j$  indexes the two-qubit gates in the layer and $N_g$ is the total number of such gates in a layer. So
\begin{equation}
\begin{aligned}
\mathbb{E}_{C_{i}}{\rm Tr}(s_i\mathcal{C}_i(s_{i-1}))^{2}&=\bigotimes_j^{N_g}\mathbb{E}_{C_{i}^{(j)}}\left({\rm{Tr}}(s_i^{(j)}s_i^{(j+1)}C_{i}^{(j)}s_{i-1}^{(j)}s_{i-1}^{(j
+1)}C_{i}^{(j)\dagger})\right)^2
\\&\geq (\frac{1}{15})^{|s_i|}.
\label{eq:layer bound}
\end{aligned}
\end{equation}
For a legal path, each two-qubit gate block in layer $i$ is either inactive (input and output are both $II$) or active (input and output are both non-$II$).By the local 2-design identity, every active block contributes a factor $1/15$ to the layer-averaged squared overlap, while inactive blocks contribute
$1$. The last line is from the fact that the number of active blocks is at most the Hamming weight $|s_i|$  of the output Pauli string. Therefore, we have
\begin{equation}
    \mathbb{E}_{\mathcal{C}} (L(\mathcal{C},s))\geq (1/15)^{|s|}.
\end{equation}
\end{proof}

\begin{lemma}[number of legal Pauli paths]
    For a $d$-depth random quantum circuit $\mathcal{C}$, the number of legal Pauli paths with weight less than $l$ is upper bounded by $2^{\mathcal{O}(l')}$.
    \label{le:path number}
\end{lemma}
\begin{proof}
Let $N_s$ denote the number of legal Pauli paths with weight $l$, we have
\begin{equation}
    \begin{aligned}
    \mathcal{O}(1)&=\sum_{|s|=d+1}^l \mathbb{E}_{\mathcal{C}}(L(\mathcal{C},s))\\
        &\geq  \sum_{|s|=d+1}^l (\frac{1}{15})^{|s|}+1\\
        &\geq  \sum_{|s|=d+1}^l (\frac{1}{15})^{l}+1\\
        & =  (\frac{1}{15})^{l}N_{|s|\in[d+1,l]}+1.
    \end{aligned}
\end{equation}

where $N_{|s|\in[d+1,l]}$ denotes the legal Pauli paths except all identity one. The number of Pauli paths needed is 
\begin{equation}
\label{equ:pathn}
    N_{s} =N_{|s|\in[d+1,l]}+1 = \mathcal{O}(1) 15^l=2^{\mathcal{O}(l)}=2^{\mathcal{O}(l'+d)}.
\end{equation}
In our learning framework, we focus on the terminal Pauli operators $s_d$ rather than the internal path trajectories. Since the collective influence of all contributing paths is implicitly encapsulated within the learned coefficients $\{\alpha_P\}$, the complexity is governed by the Hamming weight of the output operators, $l'$. To establish the relationship between $l$ and $l'$, we note that any non-trivial Pauli path must account for at least one non-identity operation per layer to remain active during the $d$-layer evolution, yielding $W_k=0$ for $1 \le k \le d$. This implies that for a total path weight $l$, the remaining degrees of freedom for the terminal operator weight is roughly $l' = l - d$.
\end{proof}
While our algorithm estimates the coefficients by tomographic measurements rather than by explicitly summing over legal paths, each learned coefficient $\alpha_P$ represents the aggregate contribution of all internal Pauli histories that terminate at the same Pauli operator $P$. Hence the multiplicity of internal paths affects the value of the learned parameter, but it does not introduce additional parameters to be learned. For a path cutoff $l$, define the path-induced terminal Pauli set
\begin{equation}
    \mathcal{P}_{\rm term}^{(l)}
    :=
    \left\{P\in\mathcal{P}_n:\exists~\text{legal path }s~\text{with } |s|\leq l~\text{and }s_d=\tilde{P}\right\},\quad \tilde{P}=2^{-n/2}P
\end{equation}
and denote its cardinality by $N_P(l):=|\mathcal{P}_{\rm term}^{(l)}|$. The relevant sparsity parameter is therefore the number of distinct terminal Pauli operators in $\mathcal{P}_{\rm term}^{(l)}$, rather than the number of internal path trajectories.
\begin{lemma}
Consider the noisy random circuit $\mathcal{C}$ in Eq.~\ref{Eq:noisychannel} with constant noise strength $\gamma=\Theta(1)$. In the nontrivial learning regime, there exists a truncated Pauli expansion
\begin{equation}
    \hat{\rho}=\sum_{P\in\mathcal{P}_{\rm term}^{(l)}} \alpha_P P
\end{equation}
such that $\|\hat{\rho}-\rho\|_2\leq \epsilon$. For inverse-polynomial accuracy, the number of relevant path-induced terminal coefficients is bounded by
\begin{equation}
    N_P(l):=|\mathcal{P}_{\rm term}^{(l)}|\leq 2^{\mathcal{O}(l')},
\end{equation}
where $l'=l-d=\mathcal{O}(\log(1/\epsilon))$ in the nontrivial depth regime.
A naive enumeration of all Pauli operators with weight at most $l'$ over the full $n$-qubit system would instead scale as $\sum_{k=0}^{l'}\binom{n}{k}3^k=n^{\mathcal{O}(l')}$.
\end{lemma}
\begin{proof}
    Let $N_s(l)$ denote the number of relevant legal Pauli paths with $|s|\leq l$, and let $N_P(l)$ denote the number of distinct terminal Pauli operators in $\mathcal{P}_{\rm term}^{(l)}$. The apparent depth dependence in the path count is
    \begin{equation}
        N_s(l)\leq 2^{\mathcal{O}(l'+d)}.
    \end{equation}
    This exponential dependence on $d$ does not undermine the learning advantage. Indeed, for constant noise strength $\gamma=\Theta(1)$, sufficiently deep noisy circuits are already driven close to the trivial fixed point of the noise, namely the maximally mixed state in the unital case. Once the depth is large enough that this trivial approximation is within the target precision, the learning task can be solved by directly outputting the trivial state and no Pauli-support enumeration is needed. Therefore, the only nontrivial regime to analyze is the depth range before this trivialization occurs.

    More generally, the nontrivial depth range before this fixed-point approximation becomes accurate is logarithmic in both the system size and the inverse target accuracy, $d=\mathcal{O}(\log n+\log(1/\epsilon))$. For inverse-polynomial target accuracy, $\epsilon=n^{-\Theta(1)}$, this reduces to $d=\mathcal{O}(\log n)$. Choosing the cutoff $l'$ proportional to $\log(1/\epsilon)$ then gives $l'=\Theta(\log n)$ in this regime, and hence $d=\mathcal{O}(l')$. Substituting this into Eq.~\ref{equ:pathn} gives
    \begin{equation}
        N_s(l)\leq 2^{\mathcal{O}(l'+d)}=2^{\mathcal{O}(l')}.
    \end{equation}
    This is a bound on the number of internal path trajectories, not yet on the number of learned parameters. Our algorithm does not learn each internal Pauli path separately. All legal paths terminating at the same Pauli operator $P=s_d$ are absorbed into a single coefficient $\alpha_P$. Therefore the terminal count satisfies
    \begin{equation}
        N_P(l)\leq N_s(l)\leq 2^{\mathcal{O}(l')}.
    \end{equation}
    This proves the claimed support-size bound for the retained terminal Pauli operators in the nontrivial learning regime.
\end{proof}
\subsection{Applications of the Learned State Representation}
\label{sec:qst-applications}
The output of the state-learning routine is the sparse Pauli representation
\begin{equation}
    \hat\rho=\sum_{P\in\mathcal T_{l'}}\hat\alpha_P P,
    \qquad
    \mathcal T_{l'}=\{P\in\mathcal P_n: |P|\leq l'\}.
\end{equation}
This representation should be understood as a classical proxy for the noisy state in the Schatten-$2$ (Frobenius) norm, rather than as a full trace-distance tomographic reconstruction. It is unnecessary to store the full $2^n\times 2^n$ density matrix; downstream quantities that are stable under such an $L_2$ perturbation can be evaluated directly from the retained Pauli coefficients. In particular, for any test operator $A$ with bounded Hilbert--Schmidt norm,
\begin{equation}
    \left|
    {\rm Tr}\!\left[A(\rho-\hat\rho)\right]
    \right|
    \leq
    \|A\|_2\|\rho-\hat\rho\|_2 .
    \label{eq:qst-application-hs}
\end{equation}
Thus an $\epsilon$-accurate Frobenius approximation yields an additive prediction guarantee of at most $\epsilon\|A\|_2$ for all such probes.

The direct applicability of Eq.~\eqref{eq:qst-application-hs} is contingent on the Hilbert--Schmidt norm of $A$. For low-weight Pauli operators or $k$-local observables acting on a constant number of qubits, $\|A\|_2$ grows at most as $2^{k/2}$, independently of the total system size $n$; in this regime the Frobenius guarantee furnishes a system-size-independent additive error bound. Specifically, if $A$ is a Pauli string supported on $k$ qubits, then $\|A\|_2=2^{k/2}$, and Eq.~\eqref{eq:qst-application-hs} yields an error of $O(\epsilon\cdot 2^{k/2})$. When $k=O(1)$ this bound is constant-order, rendering the low-weight Pauli representation an efficient classical data structure for predicting local physical quantities~
\cite{huangPredictingManyProperties2020}.

An immediate application of this framework is pure-state overlap estimation. If $A=\sigma=|\psi\rangle\langle\psi|$ is a known reference pure state, then $\|\sigma\|_2=1$, and ${\rm Tr}(\sigma\hat\rho)$ estimates the overlap with the noisy output to additive error at most $\epsilon$. This provides a simple route to pure-state verification whenever the reference state admits an efficient classical Pauli description or can otherwise be queried~
\cite{flammiaDirectFidelityEstimation2011}. For mixed reference states the same formula yields an estimate of the Hilbert--Schmidt inner product; it should be noted that this quantity is not in general identical to the fidelity $F(\rho,\sigma)=\|\sqrt{\rho}\sqrt{\sigma}\|_1^2$.

Another natural class of probes is furnished by $k$-local observables. For a $k$-local operator with Pauli expansion $A=\sum_P a_P P$, the quantity ${\rm Tr}(A\hat\rho)$ can be evaluated as an inner product of the Pauli coefficient vectors under the appropriate normalization convention. Since the Hilbert--Schmidt norm of a $k$-local Pauli operator is at most $2^{k/2}=O(1)$, the Frobenius guarantee directly yields an $O(\epsilon)$ additive error. Such observables include local spin correlations, short-range entanglement probes, and other physically common test operators~
\cite{huangPredictingManyProperties2020}. For non-local observables with large Hilbert--Schmidt norm, Eq.~\eqref{eq:qst-application-hs} no longer furnishes a useful additive guarantee unless they are rescaled or treated term by term.

In addition, the same representation can serve as a preprocessing object for more specialized verification tasks. For example, purity is a quadratic functional of the Pauli coefficient vector, since ${\rm Tr}(\rho^2)=\|\rho\|_2^2$. The perturbation obeys
\begin{equation}
    \left|{\rm Tr}(\rho^2)-{\rm Tr}(\hat\rho^2)\right|
    \leq
    \|\rho-\hat\rho\|_2\left(\|\rho\|_2+\|\hat\rho\|_2\right),
\end{equation}
so a controlled Frobenius approximation can be fed into dedicated purity or overlap-estimation routines~
\cite{brydgesProbingRényiEntanglement2019}. These verification tasks carry their own assumptions and estimators, but the low-weight representation supplied by the noisy-state learning theorem provides the reusable classical data structure on which they operate.

\subsection{Algorithm of Learning a Quantum State}
\label{sec:sla}
For the first problem, there are several ways to get the $\hat{\rho}$. The sections following introduce 2 methods, including computing directly by classical shadow \citep{huangPredictingManyProperties2020}, and a way of learning alpha based on Ref.\cite{huangLearningShallowQuantum2024}
\subsubsection{Compute Directly}

As shown before, $\rho = \sum_{s_d\in P_n}\alpha_{s_d}s_d$, where $s_d\in\tilde{\mathcal{P}}_n$. In that case, 
\begin{equation}
    \begin{aligned}
        \alpha_{s_d}& = {\rm{Tr}}(\rho s_d)\\
        &={\rm{Tr}}(\sum_{s_d'\in\tilde{\mathcal{P}}_n}\alpha_{s_d'}s_d' s_d)\\
        & = \sum_{s_d'\in \tilde{\mathcal{P}}_n}\alpha_{s_d'}{\rm{Tr}}(s_d' s_d)\\
        & = \alpha_{s_d}.
    \end{aligned}
\end{equation}
The fourth line uses
\begin{equation}
{\rm{Tr}} (s_d s_d' ) =
\begin{cases}
0, & \text{if } s_d\neq s_d', \\
1, & \text{if } s_d = s_d'.
\end{cases}
\end{equation}
Thus, $\alpha_{s_d}$ is obtained by evaluating ${\rm{Tr}}(\rho s_d)$, where $\rho$ is estimated via classical shadows. Using a set of POVMs (Positive Operator-Valued Measures) such as the random Pauli basis that measures each qubit and yields outcomes $\ket{b}\in \{0,1\}^n$, the classical shadow is constructed as $\tilde{\rho}=\otimes^n_{j=1}\left(3P_j^{\dagger}\ket{b_j}\bra{b_j}P_j-I\right)$, immediately gives $\alpha_{s_d}={\rm{Tr}}(\tilde{\rho} s_d)$.

\subsubsection{Quantum State Tomography}











This section is mainly about a way of learning $\alpha$ based on Ref.\cite{huangLearningShallowQuantum2024}, which introduces a classical dataset to reconstruct the channel's output. Our results are given below.
\begin{theorem}[Noisy Quantum State Learning,formal]
    For any noisy quantum state $\rho$ prepared by an unital noisy quantum circuit $\mathcal{C}$ (Eq.~\ref{Equ:noisychannel}), there exists a learning algorithm that can efficiently satisfy $\|\rho-\sum_{|P|\leq l'}\alpha_P P\|_2\leq \epsilon$ with success probability $\geq 1-\delta$. The learning algorithm requires sample complexity $N_{\rm {data}}=6^{\mathcal{O}\left(\log(1/\epsilon\delta)\right)}\log(1/\delta)\epsilon^{-2}$ and classical post-processing complexity $ n\cdot24^{\mathcal{O}\left(\log(1/\epsilon\delta)\right)}\log(1/\delta)\epsilon^{-2}$.
\label{the:sla}
\end{theorem}
Details of our method are as follows.

Let ${\rm{Stab}} $ be a list of single-qubit stabilizers:
\begin{equation}
    {\rm{Stab}} = \left\{\ket{0}, \ket{1}, \ket{+},\ket{-},\ket{y+},\ket{y-} \right\}.
\end{equation}
Let $\{\ket{\psi_j}=\otimes^n_{i=1}\ket{\psi_{i,j}}\}^{N_{\rm{data}}}_{j=1}$, where $\ket{\psi_{i,j}}\in {\rm{Stab}} $.
\begin{equation}
    \begin{aligned}
        &\mathbb{E}_{\ket{\psi_j}\sim {\rm{Stab}}^{\otimes n}}\bra{\psi_j}\mathcal{C}(\ket{0^{ n}}\bra{0^{ n}})\ket{\psi_j}\bra{\psi_j}P\ket{\psi_j}\\
        &=\sum_{|P|\leq l'}\alpha_{P}\mathbb{E}_{\ket{\psi_j}\sim {\rm{Stab}}^{\otimes n}}\bra{\psi_j}P\ket{\psi_j}\bra{\psi_j}P\ket{\psi_j}\\
        &=\sum_{|P|\leq l'}\alpha_{P}\mathbb{E}_{U\sim U(2)}\bigotimes_{i=1}\bra{0}U^{\dagger}_{i,j}PU_{i,j}\ket{0}\bra{0}U^{\dagger}_{i.j}PU_{i,j}\ket{0}\\
        &= \frac{\alpha_{P}}{3^{|P|}}\bigotimes^n_{i=1}\sum_{Q\in\{X,Y,Z\}}\bra{0^2}Q\otimes Q\ket{0^2}\\
        & = \frac{\alpha_{P}}{3^{|P|}}.
    \end{aligned}
\end{equation}
The third line employs $\ket{\psi_j}=\otimes^n_{i=1}\ket{\psi_{i,j}}=\otimes^n_{i=1}U_{i,j}\ket{0}$, where $U_{i,j} \sim {\rm Cl}(2)$. The fourth line uses
\begin{equation}
    \mathbb{E}_{U_{i,j}\sim {\rm Cl}(2)}\left[U_{i,j}^{\dagger\otimes 2}(Q_i\otimes Q_i')U_{i,j}^{\otimes 2}\right]=
    \left\{
    \begin{aligned}
       & I^{\otimes 2}, \  \text{if} \ Q_i=Q_i'=I,\\
       & \frac{1}{3}\sum_{Q_i\in\{X,Y,Z\}^{\otimes 2}}\left(Q_i\otimes Q_i\right), \ \text{if} \ Q_i=Q_i'\neq I,\\
       & 0, \ \text{if} \ Q_i\neq Q_i'.
        \end{aligned}
    \right.
\label{Eq:cliffordproperty appendix}    
\end{equation}

Therefore, $\alpha_{P}$ can be calculated by
\begin{equation}
    \begin{aligned}
        \alpha_{P}  &=3^{|P|}\mathbb{E}_{\ket{\psi_j}\sim {\rm{Stab}}^{\otimes n}}\bra{\psi_j}\rho\ket{\psi_j}\bra{\psi_j}P\ket{\psi_j}\\
        &\approx  \frac{3^{|P|}}{N_{\rm{data}}}\sum_{j=1}^{N_{\rm{data}}}\bra{\psi_j}\rho\ket{\psi_j}\bra{\psi_j}P\ket{\psi_j}.
    \end{aligned}
\end{equation}
The first part of the summation term (of the form $\bra{\psi_i}\rho\ket{\psi_i}$) can be obtained by using the SWAP-test method, while the latter part can be derived through classical post-processing. The data complexity $N_{\rm{data}}$ is 
$6^{\mathcal{O}\left(l'\right)}\epsilon^{-2}\log(1/\delta),$
with failure probability $\delta$. The proof follows a similar logic to that in Appendix~\ref{sec:pt2}, employing the triangle inequality and Hoeffding’s inequality. The classical post-processing complexity is $l'\mathcal{O}(n\cdot N_{\rm{data}} \cdot 2^{l'})$ according to the procedure presented in Algorithm~\ref{alg:qsl} and Lemma~\ref{le:path number}.











\section{Learning a Quantum Process Characterization}

Compared with the noisy quantum state tomography, QPT is a more challenging task, which requires an exponential query complexity in the worst-case scenario Ref.~\cite{haahQueryoptimalEstimationUnitary2023}, rendering it infeasible for large-scale systems. Inspired by the noisy quantum state tomography method, we proposed an efficient learning algorithm for the QPT task, particularly when the quantum process is given by a noisy quantum circuit $\mathcal{C}$~(Eq.~\ref{Eq:noisychannel}) followed by an unknown quantum measurement $O$. Without loss of generality, we assume the an $n$-qubit observable $O=\sum_{k=1}^M c_k Q_k$ with $M={\rm poly}(n)$, is the linear combinations of operators. Here, we do not restrict the locality of $O$, actually, it may include global Pauli terms with $\abs{Q_k}=\mathcal{O}(n)$.

Let the noisy quantum channel be given by the Kraus decomposition  $\mathcal{C}=\sum_{j} K_j(\cdot)K_j^{\dagger}$. It is observed that 
\begin{align}
    {\rm Tr}\left[\mathcal{C}(\rho_{\rm in})O\right]={\rm Tr}\left[\sum_{j} K_j\rho_{\rm in}K_j^{\dagger}O\right]={\rm Tr}\left[\sum_{j} \rho_{\rm in}K_j^{\dagger}OK_j\right]= {\rm Tr}\left[\rho_{\rm in}\mathcal{C}^{\dagger}(O)\right].
\end{align}
Consequently, the key step is to learn the `dual' representation $\mathcal{C}^{\dagger}(O)$. We demonstrate that this dual operator also admits low-weight Pauli paths, allowing for a truncation-based approximation similar to that employed for noisy states. 
\subsection{Non-unital Noise}
\label{sec:nu}
Ref.~\cite{angrisaniSimulatingQuantumCircuits2025} gives a way of simulating arbitrary noise by Pauli propagation. Generally, a non-unital noise single-qubit channel $\mathcal{E}$ can be decomposed as
\begin{equation}
    \begin{aligned}
        \mathcal{E} = \mathcal{E}_{depo}^\gamma  \circ\mathcal{E}',
    \end{aligned}
\end{equation}
where $\mathcal{E}'$ is a suitable (non-physical) linear map and $\mathcal{E}_{depo}^\gamma$ is a depolarizing noise with the effective depolarizing rate $\gamma = 1-\chi_{\mathcal{D}}(\mathcal{E})$:
\begin{equation}
     \chi_{\mathcal{D}}^2(\mathcal{E}) := \max_{\mathcal{I} \subseteq [n]} \max_{\substack{\|O_{\mathcal{I}}\|_{\rm F} \neq 0 \\ {\rm{supp}}(O_{\mathcal{I}}) = \mathcal{I}}} \left(  \frac{\| \mathcal{E}^{\dagger \otimes n} O_{\mathcal{I}} \|_{\rm F}^2}{\|O_{\mathcal{I}} \|_{\rm F}^2}  \right)^{1/|\mathcal{I}|}
\end{equation}
is the mean squared contraction coefficient of $\mathcal{E}$ in terms of the locally unbiased distribution $\mathcal{D}$. Given an observable $O=\sum_{P\in \mathcal{P}_n}\alpha_{P} P$, $O_{\mathcal{I}}$ retains those Pauli terms whose support is exactly $\mathcal{I}$: nontrivial on $\mathcal{I}$ and identity elsewhere. $|\mathcal{I}|$ is the size of the ${\rm{supp}}(O_{\mathcal{I}})$. 

Channel $\mathcal{E}'$ has the following property

\begin{equation}
    \begin{aligned}
        \mathcal{E}'(I)&=\frac{\mathcal{E}(I)}{1-\gamma}-\frac{Ip}{1-\gamma},\\
        \mathcal{E}'(X)&=\frac{\mathcal{E}(X)}{1-\gamma},\\
        \mathcal{E}'(Y)&=\frac{\mathcal{E}(Y)}{1-\gamma},\\
        \mathcal{E}'(Z)&=\frac{\mathcal{E}(Z)}{1-\gamma}.
    \end{aligned}
\end{equation}
Moreover, the adjoint channel $\mathcal{E}'^\dagger$ satisfies
\begin{equation}
    \begin{aligned}
         \mathcal{E}'^\dagger(I)&=I,\\
        \mathcal{E}'^\dagger(X)&=\frac{\mathcal{E}^\dagger(X)}{1-\gamma},\\
        \mathcal{E}'^\dagger(Y)&=\frac{\mathcal{E}^\dagger(Y)}{1-\gamma},\\
        \mathcal{E}'^\dagger(Z)&=\frac{\mathcal{E}^\dagger(Z)}{1-\gamma}.
    \end{aligned}
    \label{equ:property of E'}
\end{equation}


In that case, the output of noisy quantum circuits with non-unital channel can be written as
\begin{equation}
    \begin{aligned}
        \mathcal{C}^{\dagger}(O) &= \sum_{s\in{\tilde{\mathcal{P}}_n^{d+1}}}s_0
        {\rm{Tr}}(s_1\mathcal{C}_1^\dagger\mathcal{E}^{\otimes n \dagger}(s_{0}))
\cdots{\rm{Tr}}(s_d\mathcal{C}_d^\dagger\mathcal{E}^{\otimes n \dagger}(s_{d-1})){\rm{Tr}}(s_dO)\\
        &=\sum_{s\in{\tilde{\mathcal{P}}_n^{d+1}}}(1-\gamma)^{|s|}s_0\Phi(\mathcal{C}^\dagger,s),
    \end{aligned}
\end{equation}
where $\Phi(\mathcal{C}^\dagger,s)={\rm{Tr}}(s_1\mathcal{C}_1^\dagger\mathcal{E'}^{\otimes n \dagger}(s_{0}))
\cdots{\rm{Tr}}(s_d\mathcal{C}_d^\dagger\mathcal{E'}^{\otimes n \dagger}(s_{d-1})){\rm{Tr}}(s_dO)$.
Lemma 10 of Ref.~\cite{angrisaniSimulatingQuantumCircuits2025} shows that for $\gamma = 1-\chi_{\mathcal{D}}(\mathcal{E})$, the (non-physical) linear map $\mathcal{E}'^\dagger$ does not increase the Frobenius norm on average.
\begin{lemma}[Non-unital Noise, Lemma 10 of Ref.~\cite{angrisaniSimulatingQuantumCircuits2025}]
    Let $\mathcal{D}$ be a $1$-design over $\mathbb{SU}(2)$ and let $\gamma = 1-\chi_{\mathcal{D}}(\mathcal{E})$. For all observables $O$, we have
    \begin{equation}
        \mathbb{E}_{V\sim\mathcal{D}^{\otimes n}}\left\|\mathcal{E}^{\prime\dagger\otimes n}(V^\dagger OV)\right\|_\mathrm{F}^2\leqslant\left\|O\right\|_\mathrm{F}^2,
    \end{equation}
    which shows the linear map $\mathcal{E}'^\dagger$ does not increase the Frobenius norm in expectation over a randomly sampled $V$.
\label{le:Fro contraction}
\end{lemma}
For an $i$-th layer of $\mathcal{C}$, $\mathcal{C}_i=V_i\circ G_i$, where $V_i\sim\mathcal{D}^{\otimes n}$ and $G_i$ acts on $\mathcal{O}(1)$ qubits.

\subsection{Under Constant noise Strength}
\label{sec:constant noise process}
In this section, we will show that when the noise strength $\gamma$ is a constant independent of the system size, the required truncation degree $l'$ scales 
\begin{lemma}[contraction of the Frobenius norm,\cite{angrisaniSimulatingQuantumCircuits2025}]
    Denoted $\tilde{\mathcal{C}}^\dagger(O)=\sum_{s\in \tilde{\mathcal{P}}^{d+1}_n}s_0 \Phi(\mathcal{C^\dagger},s)$, it holds that
    \begin{equation}
        \mathbb{E}_{\mathcal{C}}\|\tilde{\mathcal{C}}^\dagger(O)\|^2_F\leq ||O||^2_F.
    \end{equation}   
\end{lemma}

From the above lemma, one may easily verify that when noise strength is a constant, we have
\begin{equation}
    \mathbb{E}_{\mathcal{C}}\left[f(\rho_{\rm in})-{\rm Tr}(O\mathcal{C}(\rho_{\rm in}))\right]\leq \epsilon_2,
\end{equation}
for an arbitrary input state $\rho_{\rm in}$.
Thus yielding the following lemma:
\begin{lemma}
    Considering a noisy circuit defined in Equ.~\ref{Eq:noisychannel}, when the noise strength is a constant, there exists a algorithm that can satisfy
    \begin{equation}
        \left[f(\rho_{\rm in})-{\rm Tr}(O\mathcal{C}(\rho_{\rm in}))\right]\leq \epsilon_2,
    \end{equation}
    where $f(\rho_{\rm in})=\sum_{|P|\leq l}\beta_{P} {\rm Tr}(P \rho_{\rm in})$ and the number of valid $P$ is at most $2^{l'}$.
\end{lemma}
\begin{proof}
The proof is parallel to Lemma~\ref{le:pna}, with the Schrödinger-picture
state replaced by the Heisenberg-evolved observable.  Expanding
\(\mathcal C^\dagger(O)\) in Pauli paths gives
\[
\mathcal C^\dagger(O)
=
\sum_s (1-\gamma)^{|s|}\Phi(\mathcal C^\dagger,s)s_0 .
\]
Let \(\mathcal C^{(l)\dagger}(O)\) retain the paths with \(|s|\le l\).
For any input state \(\rho_{\rm in}\),
\[
\Delta(\rho_{\rm in})
:=
\left|\mathrm{Tr}\rho_{\rm in}
(\mathcal C^\dagger(O)-\mathcal C^{(l)\dagger}(O))\right|.
\]
Using the same second-moment orthogonality of random two-qubit layers as in
the state case, together with
\(|\mathrm{Tr}(\rho_{\rm in}s_0)|\le 1\) for Pauli strings, we obtain
\[
\mathbb E_{\mathcal C}\Delta(\rho_{\rm in})^2
\le
\sum_{|s|>l}(1-\gamma)^{2|s|}
\mathbb E_{\mathcal C}|\Phi(\mathcal C^\dagger,s)|^2
\le
(1-\gamma)^{2l}\|O\|_F^2 .
\]
For non-unital noise, the same bound holds after writing
\(\mathcal E=\mathcal E_{\rm depo}^\gamma\circ\mathcal E'\), because
Lemma~\ref{le:Fro contraction} ensures that \(\mathcal E'^\dagger\) does not
increase the relevant Frobenius second moment.  

Since the normalized Frobenius norm has the property that $\|O\|_F=\mathcal{O}(1)$ and $\gamma=\Omega(1)$, it suffices to choose $l=\mathcal{O}\left(\log(1/(\epsilon_2\delta_2))\right)$ to ensure that $\mathbb{E}_{\mathcal{C}}\left[f(\rho_{\rm in})-{\rm Tr}(O\mathcal{C}(\rho_{\rm in}))\right]^2\leq \epsilon_2^2$. By Markov's inequality, with probability at least $1-\delta_2$, we have $\left|f(\rho_{\rm in})-{\rm Tr}(O\mathcal{C}(\rho_{\rm in}))\right|\leq \epsilon_2$.

Focusing on the number of the legal Pauli paths, denoted $N_s$, the basic idea is to enumerate all combinations that satisfy the Definition.~\ref{def:legal path}. 
Once the non-identity positions in one layer are fixed, those in the next layer are also fixed because a legal Pauli path requires the input and the output of every gate to be either both identities or both non-identities. Starting from the first layer, the positions and count of non-identities therefore match those of the input. For a local term $Q\in O$ acting non-trivially on a constant number of qubits, $N_s$ is bound by $M2^{\mathcal{O}(l)}$.

\end{proof}

Still considering the effective depth, thus we have that when the noise strength is a constant, the Pauli combination of the legal paths is at most $M2^{\mathcal{O}(l')}$ with depth $d=\mathcal{O}(\log n)$. If $d>\mathcal{O}(\log n)$, one can directly output the zero function, which is a good approximation of the noisy circuit.

\subsection{Under Arbitrary Noise}
\label{sec: arbitrary noise process}
In this section, we show that the process-truncation argument used in the
main text remains valid for arbitrary noise strength. We keep the notation
of the main text: $\mathcal C$ denotes the noisy circuit in
Eq.~\ref{Eq:noisychannel}, $\mathcal C^\dagger(O)$ is the Heisenberg-evolved
observable, and $\rho_{\rm in}$ is an arbitrary input state. Throughout this
section, the Frobenius norm is the normalized Frobenius norm induced by the
normalized Pauli basis $\tilde{\mathcal P}_n$.

Following the low-weight Pauli propagation truncation of
Ref.~\cite{angrisaniClassicallyEstimatingObservables2025}, for an integer
$k\geq 0$, define the retained Pauli-path set
\begin{equation}
    \mathcal S_k
    :=
    \left\{
        s=(s_0,s_1,\ldots,s_d)
        \in \tilde{\mathcal P}_n^{d+1}:
        |s_t|\leq k,\quad t=1,\ldots,d
    \right\}.
\end{equation}
The corresponding path-truncated adjoint observable is
\begin{equation}
    \mathcal C_{s,k}^\dagger(O)
    :=
    \sum_{s\in\mathcal S_k}
    (1-\gamma)^{|s|}\Phi(\mathcal C^\dagger,s)s_0,
\end{equation}
where $\Phi(\mathcal C^\dagger,s)$ is the Pauli-path amplitude defined
above. For any input state $\rho_{\rm in}$, define the exact and
path-truncated prediction functions as
\begin{equation}
    f_{\mathcal C}(O;\rho_{\rm in})
    :=
    {\rm Tr}\!\left[\rho_{\rm in}\mathcal C^\dagger(O)\right],
    \qquad
    f_{s,k}(O;\rho_{\rm in})
    :=
    {\rm Tr}\!\left[\rho_{\rm in}\mathcal C_{s,k}^\dagger(O)\right].
\end{equation}
The local two-design gate ensemble considered in this work satisfies the
locally scrambling second-moment condition required by
Ref.~\cite{angrisaniClassicallyEstimatingObservables2025}.

\begin{lemma}[Low-weight Pauli propagation bound for scalar prediction, restatement of Ref.~\cite{angrisaniClassicallyEstimatingObservables2025}]
\label{lem:imported-low-weight-propagation}
Let $\gamma=0$, and assume that the unitary layers of $\mathcal C$ are
independently sampled from a locally scrambling ensemble. Then, for any input
state $\rho_{\rm in}$ satisfying $\|\rho_{\rm in}\|_1\leq 1$,
\begin{equation}
    \mathbb E_{\mathcal C}
    \left[
        \left|
            f_{\mathcal C}(O;\rho_{\rm in})
            -
            f_{s,k}(O;\rho_{\rm in})
        \right|^2
    \right]
    \leq
    \left(\frac{2}{3}\right)^{k+1}
    \|O\|_{{\rm F}}^2 .
\end{equation}
\end{lemma}

This approach builds upon prior work on low-weight Pauli propagation.
Ref.~\cite{huangLearningPredictArbitrary2023} used a $(2/3)$ decay bound but
only for inputs restricted to locally flat distributions.
Ref.~\cite{angrisaniClassicallyEstimatingObservables2025} employed
layer-wise truncation for classical simulation from a circuit description.
In contrast, our task is data-driven learning: the algorithm estimates
terminal Pauli coefficients of $\mathcal C^\dagger(O)$ directly from
input-output expectation labels and does not propagate or enumerate Pauli
paths layer by layer. We next state the noisy version used in our process
learning analysis.

\begin{lemma}[Low-weight Pauli propagation bound for scalar prediction under noise]
\label{lem:noisy low-weight-propagation}
Assume that the unitary layers of $\mathcal C$ are independently sampled from
a locally scrambling ensemble and that each layer is followed by i.i.d.
single-qubit noise of strength $\gamma$. Then, for any input state
$\rho_{\rm in}$ satisfying $\|\rho_{\rm in}\|_1\leq 1$,
\begin{equation}
    \mathbb E_{\mathcal C}
    \left[
        \left|
            f_{\mathcal C}(O;\rho_{\rm in})
            -
            f_{s,k}(O;\rho_{\rm in})
        \right|^2
    \right]
    \leq
    \left(\frac{2}{3}(1-\gamma)^2\right)^{k+1}
    \|O\|_{{\rm F}}^2 .
\end{equation}
\end{lemma}

\begin{proof}
We prove the statement by adding the noise contraction to the noiseless
low-weight Pauli propagation argument. In the Heisenberg picture, the exact
scalar prediction has the Pauli-path expansion over
$s=(s_0,s_1,\ldots,s_d)$, while $f_{s,k}(O;\rho_{\rm in})$ keeps only those
paths whose intermediate Pauli strings satisfy $|s_t|\leq k$. Therefore
$f_{\mathcal C}(O;\rho_{\rm in})-f_{s,k}(O;\rho_{\rm in})$ is precisely the
total contribution of paths that leave this low-weight sector at least once.

To express the usual layer-wise truncation without introducing a separate
projection notation, define the backward-evolved observables
\begin{equation}
    O_{d+1}:=O,
    \qquad
    O_i:=\mathcal C_i^\dagger(O_{i+1}),
    \qquad i=d,d-1,\ldots,1,
\end{equation}
and define their truncated counterparts by retaining, after each backward
layer, only Pauli strings of weight at most $k$:
\begin{equation}
    O_{d+1}^{(k)}:=O,
    \qquad
    O_i^{(k)}:=\mathcal C_i^{\dagger(k)}(O_{i+1}^{(k)}),
    \qquad i=d,d-1,\ldots,1.
\end{equation}
Here $\mathcal C_i^{\dagger(k)}$ denotes the $i$-th adjoint layer followed by
discarding all Pauli strings with weight larger than $k$. With this notation,
$f_{s,k}(O;\rho_{\rm in})={\rm Tr}[\rho_{\rm in}O_1^{(k)}]$.

For a noiseless locally scrambling layer, the standard second-moment estimate
states that the discarded part produced by the map
$O_{i+1}^{(k)}\mapsto O_i^{(k)}$ is bounded by a factor
$(2/3)^{k+1}$. This factor comes from the locally scrambling transition rule:
Pauli components of weight at least $k+1$ have exponentially suppressed
second moment. In the noisy circuit, every discarded Pauli string has weight
at least $k+1$, and the depolarizing component of the noise contributes an
extra squared factor $(1-\gamma)^{2(k+1)}$. The one-layer leakage factor is
therefore replaced by
\begin{equation}
    \left(\frac{2}{3}\right)^{k+1}(1-\gamma)^{2(k+1)}
    =
    \left(\frac{2}{3}(1-\gamma)^2\right)^{k+1}.
\end{equation}

Applying this layer-wise estimate iteratively from the output layer to the
input layer and using Pauli-path second-moment orthogonality, the discarded
path contributions add without cross terms. The standard telescoping argument
then bounds the total discarded Pauli-path weight directly by the initial
Frobenius norm of $O$, with no additional factor depending on the depth.
Consequently,
\begin{equation}
    \mathbb E_{\mathcal C}
    \left[
        \left|
            f_{\mathcal C}(O;\rho_{\rm in})-f_{s,k}(O;\rho_{\rm in})
        \right|^2
    \right]
    \leq
    \left(\frac{2}{3}(1-\gamma)^2\right)^{k+1}
    \|O\|_{{\rm F}}^2,
\end{equation}
which proves the claim.
\end{proof}

The next observation connects the path-level truncation above to the terminal
Pauli coefficients learned by our protocol.

 \begin{lemma}[Single-layer Pauli support growth]
\label{lem:single-layer-growth}
Let $\mathcal C_t$ be a depth-one unitary layer composed of disjoint one- and two-qubit gates.
For any Pauli string $P\in\mathcal P_n$, every Pauli string appearing in
the expansion of $\mathcal C_t^\dagger(P)$
has weight at most $2|P|$.
\end{lemma}

\begin{proof}
Let $\operatorname{supp}(P)$ denote the qubits on which $P$ acts nontrivially.
Since the gates in $\mathcal C_t$ act on disjoint pairs, each qubit in
$\operatorname{supp}(P)$ can only interact with at most one additional qubit
within the same layer. Therefore, the support of any Pauli string generated
by conjugating $P$ through $\mathcal C_t$ is contained in the union of gate pairs
intersecting $\operatorname{supp}(P)$, which has size at most $2|\operatorname{supp}(P)| = 2|P|$.
\end{proof}

\begin{lemma}[Terminal support induced by low-weight Pauli paths]
\label{lem:terminal-support-induced}
Consider a Heisenberg-evolved observable under the first unitary layer $\mathcal C_1$ of
disjoint one- and two-qubit gates. Let $s=(s_0,s_1)$ denote a Pauli path
truncated at weight $k$, i.e., $|s_1|\le k$.
Then every terminal Pauli string $s_0$ generated by this path satisfies $|s_0| \le 2k.$
Consequently, the set of terminal Pauli strings generated by all low-weight
paths is contained in the terminal coefficient set used by the learning algorithm:
\begin{equation}
    \{s_0: s\in \mathcal S_k\}\subseteq \mathcal T_{2k} \subseteq \mathcal T_{l'}, \quad \text{for } l'\ge 2k
\end{equation}
where
\begin{equation}
    \mathcal T_{l'} := \{ Q\in\mathcal P_n : |Q|\le l'\}.
\end{equation}
\end{lemma}

\begin{proof}
By assumption, $|s_1|\le k$. Conjugating $s_1$ through $\mathcal C_1$, which consists
of disjoint one- and two-qubit gates, can enlarge the support to at most the
union of the gate blocks intersecting $\operatorname{supp}(s_1)$. Each qubit
in the support of $s_1$ can activate at most one additional qubit. Hence the
weight of the terminal Pauli string satisfies
\begin{equation}
     |s_0| \le 2 |s_1| \le 2k.
\end{equation}
Thus \(\{s_0:s\in \mathcal{S}_k\}\subseteq \mathcal{T}_{2k}\). If \(l'\ge 2k\), then \(\mathcal{T}_{2k}\subseteq \mathcal{T}_{l'}\), proving the claim.
\end{proof}
\begin{theorem}[Terminal low-weight truncation guarantee]
\label{thm:learned-observable-bound}
Let $\mathcal C^\dagger(O)=\sum_{Q\in\mathcal P_n}\beta_QQ$ be the
Heisenberg-evolved observable of the noisy circuit $\mathcal C$. Define the
terminal low-weight truncation
\begin{equation}
    \mathcal C^{(l')\dagger}(O)
    :=
    \sum_{Q\in\mathcal T_{l'}}\beta_QQ,
    \qquad
    \mathcal T_{l'}:=\{Q\in\mathcal P_n: |Q|\leq l'\}.
\end{equation}
Then, for any input state $\rho_{\rm in}$ with $\|\rho_{\rm in}\|_1 \le 1$,
\begin{equation}
\left|{\rm Tr}\!\left[\rho_{\rm in} ( \mathcal C^{(l')\dagger}(O)-\mathcal C^\dagger(O))\right]\right| \leq \epsilon_1,
\end{equation}
with probability at least $1-\delta_1$ over the random circuit ensemble,
provided that
\begin{equation}
    l'
    =
    \mathcal O\!\left(
    \frac{\log(\|O\|_{{\rm F}}/(\epsilon_1\sqrt{\delta_1}))}
    {\log\left(3/[2(1-\gamma)^2]\right)}
    \right)=
    \mathcal O\!\left(
    \frac{\log(1/(\epsilon_1\delta_1))}
    {\log\left(3/[2(1-\gamma)^2]\right)}
    \right)..
\end{equation}
\end{theorem}

\begin{proof}
For an input state $\rho_{\rm in}$ with $\|\rho_{\rm in}\|_1\le 1$, let
$\mathcal S_k$ denote the set of retained low-weight Pauli paths whose
intermediate weights are bounded by $k$. By
Lemma~\ref{lem:terminal-support-induced}, the terminal endpoints of all
paths in $\mathcal S_k$ are contained in $\mathcal T_{2k}$. Choosing
$k=\lfloor l'/2\rfloor$ gives
\begin{equation}
    \{s_0:\,s\in \mathcal S_k\}
    \subseteq
    \mathcal T_{2k}
    \subseteq
    \mathcal T_{l'} .
\end{equation}

We now compare the terminal Pauli truncation with the path truncation.
Let
\begin{equation}
    \mathcal A_k
    :=
    \{s=(s_0,\ldots,s_d): s\notin \mathcal S_k\}
\end{equation}
be the set of paths discarded by the low-weight path truncation, and let
\begin{equation}
    \mathcal B_{l'}
    :=
    \{s=(s_0,\ldots,s_d): |s_0|>l'\}
\end{equation}
be the set of paths whose terminal endpoint is discarded by the terminal
Pauli truncation. Since every retained path \(s\in\mathcal S_k\) satisfies
\(|s_0|\le 2k\le l'\), we have
\begin{equation}
    \mathcal B_{l'}\subseteq \mathcal A_k .
\end{equation}

The terminal truncation error can be written in the Pauli-path expansion as
\begin{equation}
\operatorname{Tr}\!\left[
\rho_{\rm in}\left(\mathcal C^\dagger(O)-\mathcal C^{(l')\dagger}(O)\right)
\right]
=
\sum_{s\in \mathcal B_{l'}}
(1-\gamma)^{|s|}\Phi(\mathcal C^\dagger,s)\operatorname{Tr}[\rho_{\rm in} s_0].
\end{equation}
Here $\mathcal C^{(l')\dagger}(O)$ is the exact terminal truncation: for each
$Q\in\mathcal T_{l'}$, its coefficient $\beta_Q$ contains the sum of
all Pauli paths ending at \(Q\). Thus the only paths omitted by
$\mathcal C^{(l')\dagger}(O)$ are those with terminal endpoint $|s_0|>l'$.

By the second-moment orthogonality of distinct Pauli paths under average case~\cite{angrisaniClassicallyEstimatingObservables2025},
\begin{equation}
\begin{aligned}
&\mathbb E_{\mathcal C}
\left|
\operatorname{Tr}\!\left[
\rho_{\rm in}\left(\mathcal C^\dagger(O)-\mathcal C^{(l')\dagger}(O)\right)
\right]
\right|^2  \\
&\quad =
\sum_{s\in \mathcal B_{l'}}
\mathbb E_{\mathcal C}
\left[
(1-\gamma)^{2|s|}|\Phi(\mathcal C^\dagger,s)|^2
|\operatorname{Tr}[\rho_{\rm in} s_0]|^2
\right].
\end{aligned}
\end{equation}
Since \(\mathcal B_{l'}\subseteq \mathcal A_k\), this is bounded by
\begin{equation}
\begin{aligned}
&\sum_{s\in \mathcal A_k}
\mathbb E_{\mathcal C}
\left[
(1-\gamma)^{2|s|}|\Phi(\mathcal C^\dagger,s)|^2
|\operatorname{Tr}[\rho_{\rm in} s_0]|^2
\right] \\
&\quad =
\mathbb E_{\mathcal C}
\left|
\operatorname{Tr}\!\left[
\rho_{\rm in}\left(\mathcal C^\dagger(O)-\mathcal C_{s,k}^\dagger(O)\right)
\right]
\right|^2 .
\end{aligned}
\end{equation}
Applying Lemma~\ref{lem:noisy low-weight-propagation}, we obtain
\begin{equation}
\mathbb E_{\mathcal C}
\left|
\operatorname{Tr}\!\left[
\rho_{\rm in}\left(\mathcal C^\dagger(O)-\mathcal C^{(l')\dagger}(O)\right)
\right]
\right|^2
\le
\left(\frac{2}{3}(1-\gamma)^2\right)^{k+1}
\|O\|_{{\rm F}}^2 .
\end{equation}

Finally, Markov's inequality gives
\begin{equation}
\Pr_{\mathcal C}\!\left[
\left|
\operatorname{Tr}\!\left[
\rho_{\rm in}\left(\mathcal C^\dagger(O)-\mathcal C^{(l')\dagger}(O)\right)
\right]
\right|>\epsilon_1
\right]
\le
\frac{\left(\frac{2}{3}(1-\gamma)^2\right)^{k+1}\|O\|_{{\rm F}}^2}{\epsilon_1^2}.
\end{equation}
Choosing $k=\lfloor l'/2\rfloor$ and
\begin{equation}
    l'
=
\mathcal O\!\left(
\frac{\log(\|O\|_{{\rm F}}/(\epsilon_1\sqrt{\delta_1}))}
{\log\left(3/[2(1-\gamma)^2]\right)}
\right)
\end{equation}
makes the right-hand side at most $\delta_1$, up to a change of the
absolute constant in the $\mathcal O(\cdot)$. This proves the claimed
high-probability terminal truncation guarantee.
\end{proof}

\begin{lemma}[Number of terminal Pauli coefficients]
\label{lem:terminal-Pauli-count}
Let $\mathcal T_{l'}$ denote the set of all $n$-qubit Pauli strings of weight at most $l'$. Then
\begin{equation}
    N_P=|\mathcal T_{l'}| = \sum_{r=0}^{l'} 3^r \binom{n}{r} = n^{\mathcal O(l')}.
\end{equation}
If the circuit has a $D$-dimensional local geometry and the observable $O$ has constant-size initial support, the relevant terminal strings lie within the light cone of $O$, yielding
\begin{equation}
    |\mathcal T_{l'}^{\rm LC}| \le \sum_{r=0}^{l'} 3^r \binom{M_{\rm LC}}{r} \le (O(d^D))^{l'} ,
\end{equation}
where $M_{\rm LC}$ is the number of qubits in the light cone of $O$.
\end{lemma}

\begin{proof}
Pauli paths are used only as an analytical representation; the learning algorithm estimates terminal coefficients directly. Each Pauli string of weight $r$ can be chosen by selecting $r$ qubits out of $n$ and assigning to each qubit one of the three non-identity Paulis, giving $3^r \binom{n}{r}$ possibilities. Summing over all $r\le l'$ yields the total number of strings.  

For a geometrically local circuit, only qubits in the light cone of $O$ can influence the terminal Pauli strings. Let $M_{\rm LC}$ be the number of qubits in the light cone. Then the number of possible weight-$r$ Pauli strings is at most $3^r \binom{M_{\rm LC}}{r}$, giving the bound on $|\mathcal T_{l'}^{\rm LC}|$.
\end{proof}
\begin{remark}
    If additional geometric locality is available, the counting can be further restricted to the light cone of \(O\). 
\end{remark}

\subsection{ Algorithm of process learning}
Similar to the noisy quantum state tomography task, reconstructing $\mathcal{C}^{(l')\dagger}(O)$ proceeds from the data set $\mathcal{D}_{\rm QPT}=\left\{\ket{\psi_j}=\otimes^{n}_{i=1}\ket{\psi_{i,j}}, \phi_j = {\rm Tr}\left[O\mathcal{C}(|\psi_j\rangle\langle\psi_j|)\right]\right\}_{j=1}^{N_{\rm data}}$, where $\ket{\psi_{i,j}}$ is a single-qubit stabilizer randomly sampled from the set ${\rm Stab}$, and $\phi_{j}$ denotes the output of the target quantum process. According to the Eq.~\ref{Eq:cliffordproperty}, coefficients $\beta_P$ can be learned efficiently via
\begin{align}
    \beta_P =\frac{3^{|P|}}{N_{\rm data}}\sum_{i=1}^{N_{\rm data}}\phi_j\bra{\psi_i}P\ket{\psi_i}.
\end{align}

The complete QPT procedure is summarized in Algorithm~\ref{alg:qpla}.
\begin{algorithm}
\label{Algorithm}
\caption{Quantum Process Learning Algorithm}
\label{alg:qpla}
\textbf{Input:} Data set $\mathcal{D}_{\rm QPT}=\left\{\ket{\psi_j}=\otimes^{n}_{i=1}\ket{\psi_{i,j}}\right\}_{j=1}^{N_{\rm data}}$ and accuracy parameter $\epsilon$;

\textbf{Output:} A $f(\cdot)$ such that $\abs{f(\cdot)-{\rm Tr}\left[O\mathcal{C}(\cdot)\right]}\leq\epsilon$ with high success probability for all input quantum states;

Let $l^{\prime}=[\log(1/\epsilon)]$, enumerate all the $P\in \mathcal{P}_n$ with $|P|\leq l^{\prime}$;

\textbf{For} $j\in[N_{\rm data}]$:

\quad Take the input state $|\psi_j\rangle\langle\psi_j|$ into the target quantum process, and obtain the output $\phi_j={\rm Tr}\left[O\mathcal{C}(|\psi_j\rangle\langle\psi_j|)\right]$;


\textbf{End For}

\textbf{For} each $P\in \mathcal{P}_n$ with $|P|\leq l^{\prime}$:

\quad Compute $\beta_{P}   =\frac{3^{|P|}}{N_{\rm data}}\sum_{j=1}^{N_{\rm data}}\phi_j\bra{\psi_j}P\ket{\psi_j}$.

\textbf{End For}


\textbf{Output}: $f(\cdot) = \sum_{|P|\leq l^{\prime}}\beta_{P}{\rm Tr}(P(\cdot))$
\end{algorithm}

\subsection{Sample Complexity and Runtime Complexity}
\label{sec:pt2}
In this section, we will prove the main result of our learning algorithm. 
\begin{theorem}[Noisy Quantum Process Learning, formal]
\label{the:prlaa}
    For any noisy quantum process $\mathcal{C}$ defined as Eq.~\ref{Equ:noisychannel}, where $\mathcal{C}_i$ is a layer of two-qubit Haar random quantum gates, and an $n$-qubit observable $O=\sum_{k=1}^M c_k Q_k$ with $M={\rm poly}(n)$, there exists a learning algorithm  $f$, such that for any $\epsilon\in(0,1)$ and input quantum state $\rho_{\rm in}$, $|f(\rho_{\rm in})-{\rm Tr}(O\mathcal{C}(\rho_{\rm in}))|<\epsilon$with success probability $\geq 1-\delta$. Define $l'_1 = \mathcal{O}\left(\log(1/ \epsilon \delta)\right)$  and $ l'_2 = \mathcal O\!\left(
    \frac{\log(1/(\epsilon\delta))}
    {\log\left(3/[2(1-\gamma)^2]\right)}
    \right)$. The learning algorithm requires sample complexity 
\begin{equation}
N_{\rm data}=
\begin{cases}
    \displaystyle
    l'_1\frac{6^{\mathcal O(l'_1)}}{\epsilon^2}
    \log\frac{1}{\delta},
    & \text{constant noise and logarithmic depth},\\[1.2em]
    \displaystyle
    \frac{3^{\mathcal O(l')}n^{\mathcal O(l'_2)}}{\epsilon^2}
    \log\frac{n^{\mathcal O(l'_2)}}{\delta},
    & \text{arbitrary noise},\\[1.2em]
    \displaystyle
    \frac{3^{\mathcal O(l')}(\mathcal O(d^D))^{l'_2}}{\epsilon^2}
    \log\frac{(\mathcal O(d^D))^{l'_2}}{\delta},
    & \text{arbitrary noise with local connectivity}.
\end{cases}
\end{equation}
    and classical post-processing complexity the total classical runtime complexity is:
\begin{equation}
T
=
\mathcal{O}\left(n \cdot N_{\rm data} \cdot N_P\right)
=
\begin{cases}
    \displaystyle
    \mathcal{O}\!\left(
    n\cdot l'_1 6^{\mathcal O(l'_1)}
    \epsilon^{-2}\log\frac{1}{\delta}
    \right),
    & \text{constant noise and logarithmic depth},\\[1.2em]
    \displaystyle
    \mathcal{O}\!\left(
    n\cdot 3^{\mathcal O(l'_2)}n^{\mathcal O(l'_2)}
    \epsilon^{-2}\log\frac{n^{\mathcal O(l'_2)}}{\delta}
    \right),
    & \text{arbitrary noise},\\[1.2em]
    \displaystyle
    \mathcal{O}\!\left(
    n\cdot 3^{\mathcal O(l'_2)}(\mathcal O(d^D))^{l'_2}
    \epsilon^{-2}\log\frac{(\mathcal O(d^D))^{l'_2}}{\delta}
    \right),
    & \text{arbitrary noise with local connectivity}.
\end{cases}
\end{equation}

    \label{the:pla}
\end{theorem}
\begin{proof}
The discrepancy between the algorithm's learned outcome and the true value, quantified via absolute value, encompasses two types of errors: truncation error and learning error.
\begin{equation}
\begin{aligned}
     &\left|f(\rho_{\rm in})-{\rm{Tr}}(\mathcal{C}(\rho_{\rm in})O\right|
     =\left|\sum_{|P|\leq l^\prime}\hat{\beta}_{P}{\rm{Tr}}(\rho_{\rm in}P) - {\rm{Tr}}(\mathcal{C}(\rho_{\rm in})O) \right| 
    \\\leq
    &\left|\sum_{|P|\leq l^\prime}\beta_{P} {\rm{Tr}}(\rho_{\rm in}P) - {\rm{Tr}}(\mathcal{C}(\rho_{\rm in})O)) \right|+\left|\sum_{|P|\leq l^\prime}\hat{\beta}_{P}{\rm{Tr}}(\rho_{\rm in}P) - \sum_{|P|\leq l^\prime}\beta_P {\rm{Tr}}(\rho_{\rm in}P) \right|\\
    &=\left|\sum_{|\tilde{P}|\leq l^\prime}\beta_{\tilde{P}} {\rm{Tr}}(\rho_{\rm in}\tilde{P}) - {\rm{Tr}}(\mathcal{C}(\rho_{\rm in})O)) \right|+\left|\sum_{|P|\leq l^\prime}\hat{\beta}_{P}{\rm{Tr}}(\rho_{\rm in}P) - \sum_{|P|\leq l^\prime}\beta_P {\rm{Tr}}(\rho_{\rm in}P) \right|. 
\end{aligned}
\end{equation}
The inequality is derived through the application of the triangle inequality, where the first term on the right-hand side of the inequality represents the truncation error, and the second term represents the learning error. $\hat{\beta}_{P}$ denotes the learned value of $\beta_{P}$.

The proof for the truncation error can be analogously extended from that in Appendix~\ref{sec:constant noise process} and~\ref{sec: arbitrary noise process}, demonstrating that when $l'= \mathcal{O}\left( \frac{\log(\|O\|_{\rm F}^2 / \epsilon_2^2 \delta_2)}{\log \left( \Lambda (1-\gamma)^{-2} \right)} \right)$, $\left|\sum_{|P|\leq l'}\beta_P {\rm Tr}(\rho_xP) - {\rm Tr}(\mathcal{C}(\rho_x)O) \right|\leq \epsilon_2$.

The learning error is bounded by
\begin{equation}
\label{equ:abound}
   \begin{aligned}
       \left|\sum_{|P|\leq l'}\left(\hat{\beta}_P - \beta_P\right){\rm Tr}(\rho_{\rm in}P) \right|
       &\leq \left|\sum_{|P|\leq l'}\hat{\beta}_P - \sum_{|P|\leq l'}\beta_P\right|\\
       &= \sum_{|P|\leq l'}\left|\hat{\beta}_P - \beta_P\right|\\
       &\leq N_P \max_{|P|\leq l'}|\hat{\beta}_P - \beta_P|\\
       &\leq \epsilon_3.
   \end{aligned}
\end{equation}
Combining the equation with Hoeffding's inequality, we can derive that given a dataset of size $N_{\rm{data}}= \frac{3^{\mathcal{O}(l')}N_P^2}{\epsilon^2_3}\log\frac{N_P}{\delta}$ 
with probability at least $1-\delta$, Eq.~\ref{equ:abound} is valid. 
Given $\epsilon_2=\epsilon_3=\Theta(\epsilon)$, substituting the corresponding terminal coefficient count $N_P$ gives
\begin{equation}
N_{\rm data}=
\begin{cases}
    \displaystyle
    l'_1\frac{6^{\mathcal O(l'_1)}}{\epsilon^2}
    \log\frac{1}{\delta},
    & \text{constant noise and logarithmic depth},\\[1.2em]
    \displaystyle
    \frac{3^{\mathcal O(l')}n^{\mathcal O(l'_2)}}{\epsilon^2}
    \log\frac{n^{\mathcal O(l'_2)}}{\delta},
    & \text{arbitrary noise},\\[1.2em]
    \displaystyle
    \frac{3^{\mathcal O(l')}(\mathcal O(d^D))^{l'_2}}{\epsilon^2}
    \log\frac{(\mathcal O(d^D))^{l'_2}}{\delta},
    & \text{arbitrary noise with local connectivity}.
\end{cases}
\end{equation}

For the runtime complexity of classical post-processing, the calculation is derived directly from Algorithm \ref{alg:qpla}. The dominant factor in the runtime is the computation of the coefficients $\beta_{P}$, which involves nested iterations over $N_{\rm{data}}$ input samples and $N_P$ terminal Pauli strings. The internal calculation of the expectation value $\langle\psi_j|P|\psi_j\rangle$ has a cost of $\mathcal{O}(n)$, because the input state $\ket{\psi_j}$ is a product state and $P$ is a Pauli string, allowing the expectation value to be computed via the product of $n$ single-qubit terms.\\
Thus, the total classical runtime complexity is:
\begin{equation}
T
=
\mathcal{O}\left(n \cdot N_{\rm data} \cdot N_P\right)
=
\begin{cases}
    \displaystyle
    \mathcal{O}\!\left(
    n\cdot l'_1 6^{\mathcal O(l'_1)}
    \epsilon^{-2}\log\frac{1}{\delta}
    \right),
    & \text{constant noise and logarithmic depth},\\[1.2em]
    \displaystyle
    \mathcal{O}\!\left(
    n\cdot 3^{\mathcal O(l'_2)}n^{\mathcal O(l'_2)}
    \epsilon^{-2}\log\frac{n^{\mathcal O(l'_2)}}{\delta}
    \right),
    & \text{arbitrary noise},\\[1.2em]
    \displaystyle
    \mathcal{O}\!\left(
    n\cdot 3^{\mathcal O(l'_2)}(\mathcal O(d^D))^{l'_2}
    \epsilon^{-2}\log\frac{(\mathcal O(d^D))^{l'_2}}{\delta}
    \right),
    & \text{arbitrary noise with local connectivity}.
\end{cases}
\end{equation}
where the factor $n$ accounts for the linear cost of evaluating the $n$ single-qubit terms that constitute $\langle\psi_j|P|\psi_j\rangle$ for each pair of sample $\ket{\psi_j}$ and Pauli string $P$.\\
Thus the complexity statement is separated into three regimes. For constant noise strength and constant accuracy, $l'$ is constant in the nontrivial regime and $N_P\leq 2^{\mathcal O(l')}$. For arbitrary noise strength, direct enumeration gives $N_P=|\mathcal T_{l'}|=n^{\mathcal O(l')}$, which is quasi-polynomial for inverse-polynomial accuracy. If the circuit has a $D$-dimensional local geometry and $O$ has constant-size support, the enumeration can be restricted to the light cone and $N_P\leq |\mathcal T_{l'}^{\rm LC}|\leq(\mathcal O(d^D))^{l'}$.
\end{proof}

\section{Sample Complexity Lower bound for the worst-case scenario}
\label{Sec:samplelowerbound}

The main manuscript essentially considers learning an efficient classical representation of noisy quantum states and processes in the average-case scenario. As we claimed in Theorems~\ref{the:sla} and~\ref{the:pla}, the tasks of learning noisy quantum states and performing tomography are highly efficient in the average-case setting. However, this does not rule out intrinsic hardness in the worst case. Here we theoretically demonstrate that learning noisy quantum states prepared by quantum circuits subject to constant-strength noise channels is quantum-hard in the worst-case scenario.

The fundamental idea relies on constructing a polynomial reduction to the quantum state discrimination problem.

\vspace{5px}

\begin{task}
    Consider two pure quantum states $\rho_0$ and $\rho_1$, and a noisy quantum circuit $\mathcal{C}$ with depth $d$, where Each quantum circuit is affected by $\gamma$-strength Pauli channel in each layer. Suppose that a distinguisher is given access to copies of the quantum states $\mathcal{C}(\rho_0)$ and $\mathcal{C}(\rho_1)$, then what is the fewest number of copies sufficing to identify these two noisy quantum states with high probability?
    \label{problem2}
\end{task}

Obviously, if one can perform quantum state tomography on these noisy states, then efficient classical representations of the noisy states are obtained. Using these classical representations, one can easily distinguish the noisy states $\mathcal{C}(\rho_0)$ from $\mathcal{C}(\rho_1)$ easily. As a result, Task~\ref{problem2} can be used to benchmark the sample-complexity lower bound for the noisy quantum state tomography problem. We state the result below.

\begin{theorem}
    Given an unknown noisy quantum state $\rho$ prepared by a $d$-depth quantum circuit affected by $\gamma$-strength local Pauli noise channels, then any algorithm designed to learn an efficient representation to $\rho$ requires at least $m$ samplings in the \emph{worst-case scenario}, where $$m=\frac{(1-\gamma)^{-2cd}(1-\eta)^2}{2n},$$ where $c=1/(2\ln 2)$ and constant $\eta\in\mathcal{O}(1)$.
    \label{them:lowerbound}
\end{theorem}

When the noise strength $\gamma=\Omega(1)$, and quantum circuit depth $d\geq {{\rm poly}\log(n)}$, the sample complexity required for quantum state tomography grows at least quasi-polynomially with the system size in the worst-case scenario. We emphasize that this result does not contradict Theorem~\ref{the:sla} and ~\ref{the:pla}: the former statement concerns the worst case, while the latter addresses the average case under the random-circuit assumption. 

In the quantum process tomography task, when $O> 0$, the target is to learn a classical representation to $\mathcal{C}^{\dagger}[O]$ which can be easily reduced to a density matrix learning task by setting $\rho=\mathcal{C}^{\dagger}[O]/{\rm Tr}[\mathcal{C}^{\dagger}[O]]$. This justifies the statement that noisy process tomography (for this observable $O$) is no easier than state tomography.

To support the proof of our result, we require the following lemmas.

\vspace{5px}

\begin{lemma}[Lemma~6 in~\cite{wang2021noise}]
    Consider a single noise channel $\mathcal{N}=\mathcal{N}_1\otimes\cdots\otimes\mathcal{N}_n$ where each local noise channel $\{\mathcal{N}_j\}_{j=1}^n$ is a Pauli noise channel that satisfies $\mathcal{N}_j(\sigma)=q_{\sigma}\sigma$ for $\sigma\in\{X,Y,Z\}$ and $q_{\sigma}$ be the Pauli strength. Then we have
    \begin{align}
        D_2\left(\mathcal{N}(\rho)\|\frac{I^{\otimes n}}{2^n}\right)\leq q^{2c}D_2\left(\rho\|\frac{I^{\otimes n}}{2^n}\right),
    \end{align}
    where $D_2(\cdot\|\cdot)$ represents the $2$-Renyi relative entropy, $q=\max_{\sigma}q_{\sigma}$ and $c=1/(2\ln 2)$.
    \label{lemma:renyiineq}
\end{lemma}

\begin{lemma}
    Given an arbitrary $n$-qubit density matrix and maximally mixed state $I^{\otimes n}/2^n$, we have
    \begin{align}
        D\left(\rho\|I^{\otimes n}/2^n\right)\leq D_2\left(\rho\|I^{\otimes n}/2^n\right),
    \end{align}
    where $D(\cdot\|\cdot)$ denotes the relative entropy and $D_2(\cdot\|\cdot)$ denotes the $2$-Renyi relative entropy.
    \label{lemma:renyi}
\end{lemma}
\emph{Proof:} Given quantum states $\rho$ and $\sigma$, the quantum $2$-Renyi entropy 
    \begin{align}
        D_2(\rho\|\sigma)=\log{\rm Tr}\left[\left(\sigma^{-1/4}\rho\sigma^{-1/4}\right)^2\right].
    \end{align}
    When $\sigma=I^{\otimes n}/2^n$, we have $D_2(\rho\|I^{\otimes n}/2^n)=\log{\rm Tr}\left[\left((I^{\otimes n}/2^n)^{-1}\rho^2\right)\right]=n+\log{\rm Tr}[\rho^2]$. Noting that the function $y=x^2-x\log x\geq 0$ when $x\in[0,1]$, and this implies ${\rm Tr}(\rho^2)\geq {\rm Tr}(\rho\log \rho)$. Finally, we have
    \begin{align}
        D\left(\rho\|I^{\otimes n}/2^n\right)=n+{\rm Tr}\left[\rho\log\rho\right]\leq {\rm Tr}\left[\rho^2\right]+n=D_2\left(\rho\|I^{\otimes n}/2^n\right).
    \end{align}
\vspace{10px}

\emph{Proof of Theorem~\ref{them:lowerbound}:}
Now we prove the sample complexity lower bound to the noisy quantum state tomography task. We consider the sample complexity $m$ in distinguishing quantum states $\mathcal{C}(\rho_0)$ and $\mathcal{C}(\rho_1)$. When their trace distance is quite large, let $\eta\in(0,1)$ and we have
\begin{equation}
\begin{aligned}
    1-\eta&\leq\frac{1}{2}\left\|\mathcal{C}(\rho_0)^{\otimes m}-\mathcal{C}(\rho_1)^{\otimes m}\right\|_1\\
    &\leq \frac{1}{2}\left(\left\|\mathcal{C}(\rho_0)^{\otimes m}-(I_n/2^n)^{\otimes m}\right\|_1+\left\|\mathcal{C}(\rho_1)^{\otimes m}-(I_n/2^n)^{\otimes m}\right\|_1\right)\\
    &\leq \frac{1}{\sqrt{2}}\left(D^{1/2}\left(\mathcal{C}(\rho_0)^{\otimes m}\|(I_n/2^n)^{\otimes m}\right)+D^{1/2}\left(\mathcal{C}(\rho_1)^{\otimes m}\|(I_n/2^n)^{\otimes m}\right)\right),
\end{aligned}
\end{equation}
where the second line comes from the triangle inequality and the third line comes from the Pinsker's inequality. Using Lemmas~\ref{lemma:renyiineq} and~\ref{lemma:renyi}, we have 
\begin{equation}
    \begin{aligned}
         1-\eta&\leq \frac{1}{\sqrt{2}}\left(D^{1/2}_2\left(\mathcal{C}^{\otimes m}(\rho_0)\|(I_n/2^n)^{\otimes m}\right)+D^{1/2}_2\left(\mathcal{C}^{\otimes m}(\rho_1)\|(I_n/2^n)^{\otimes m}\right)\right)\\
         &\leq \frac{\sqrt{nm}}{\sqrt{2}}((1-\gamma)^{cd}+(1-\gamma)^{cd})\\
         &\leq \sqrt{2nm}(1-\gamma)^{cd},
    \end{aligned}
\end{equation}

As a result we have
\begin{align}
  m\geq\frac{(1-\gamma)^{-2cd}(1-\eta)^2}{2n}.
\end{align}

\section{Additional Experiment result}

\subsection{numerical experiment for highly entangled input state}
\label{sec:highly entangled}
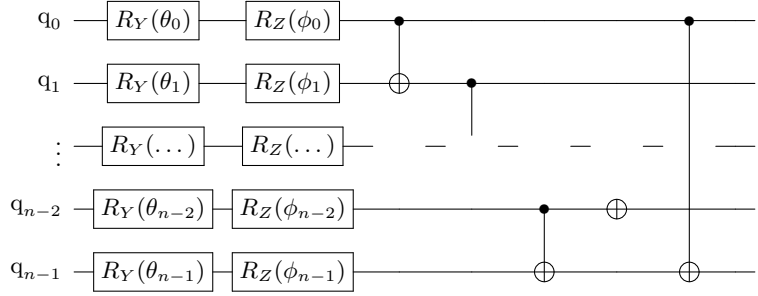
\begin{figure}[htb]
\centerline{
\Qcircuit@C=0.8em @R=1em {
    \lstick{\text{q}_{0}} & \gate{R_Y(\theta_{0})} & \gate{R_Z(\phi_{0})} & \ctrl{1} & \qw & \qw & \qw & \ctrl{4} & \qw & \qw \\
    \lstick{\text{q}_{1}} & \gate{R_Y(\theta_{1})} & \gate{R_Z(\phi_{1})} & \targ{} & \ctrl{1} & \qw & \qw & \qw & \qw & \qw \\
    \lstick{\vdots} & \gate{R_Y(\dots)} & \gate{R_Z(\dots)} & \ghost{\dots} & \ghost{\dots} & \ghost{\dots} & \ghost{\dots} & \ghost{\dots} & \qw & \qw \\
    \lstick{\text{q}_{n{-}2}} & \gate{R_Y(\theta_{n{-}2})} & \gate{R_Z(\phi_{n{-}2})} & \qw & \qw & \ctrl{1} & \targ{} & \qw & \qw & \qw \\
    \lstick{\text{q}_{n{-}1}} & \gate{R_Y(\theta_{n{-}1})} & \gate{R_Z(\phi_{n{-}1})} & \qw & \qw & \targ{} & \qw & \targ{} & \qw & \qw
}
}
\caption{The demonstration of a layer of the state preparation circuit for $n$ qubits. The structure consists of parameterized single-qubit rotations followed by a cyclic CNOT entangling layer.}
\label{fig:spc}
\end{figure}

To further underscore the novelty of our \textbf{input-agnostic regime}, we constructed a Quantum Process Tomography (QPT) experiment using an input state that is \textbf{highly entangled} and falls outside the distribution set addressed by previous work \citep{huangLearningPredictArbitrary2023}. The method detailed in \citep{huangLearningPredictArbitrary2023} requires the input distribution to be at most polynomially far from a \textbf{locally flat distribution}.

As demonstrated by  \citep{huangLearningPredictArbitrary2023}, locally flat distributions encompass: Random product states, ground and thermal states of random local Hamiltonians and any state generated by a circuit whose final layer consists of random single-qubit gates.

Here, we intentionally generated the input state $\rho_{\text{x}}$ using a two-layer parameterized circuit to ensure high entanglement. Each layer of this state preparation circuit is structured as depicted in Fig.~\ref{fig:spc}.


The experimental result of performing QPT with this highly entangled input state is presented in Fig.~\ref{fig:extra result}. The outcome clearly demonstrates that \textbf{our algorithm's performance is not constrained by the entanglement level or specific structure of the input state}, thus validating its input-agnostic nature.

\begin{figure}[h]
\begin{center}
\includegraphics[width=0.6\linewidth]{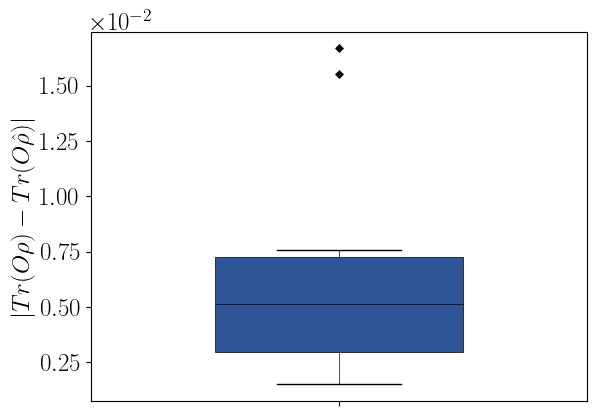} 
\end{center}
\caption{ The experiment result of QPT where the input state is generated from the form as Fig~\ref{fig:spc}. Set $l'=2$ and the process Eq.~\ref{equ:sim} with $5$ layers is at the depolarizing noise strength $0.01$. } 
\label{fig:extra result}
\end{figure}

\subsection{full matrix figure}
\label{sec:expr}
\begin{figure}[htb]
\begin{center}
\includegraphics[width=0.6\linewidth]{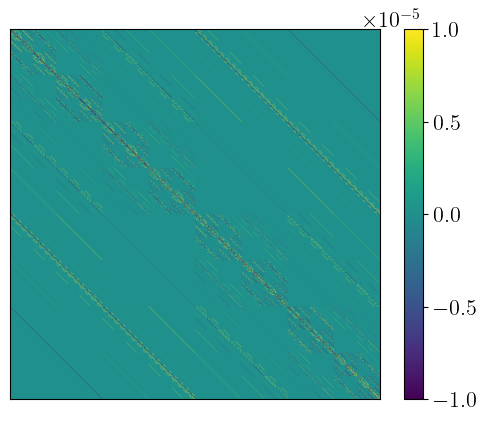} 
\end{center}
\caption{ The heatmap visualizes the full matrix of the $\tilde{\rho}-\hat{\rho}$ in Fig~\ref{fig:exp} c, when $\theta_h=\frac{\pi}{2}$. } 
\end{figure}


\end{document}